\documentclass[aps,prx,onecolumn,superscriptaddress,nofootinbib,longbibliography]{revtex4-2}
\usepackage[a4paper, left=1in, right=1in, top=0.95in, bottom=0.95in]{geometry}
\usepackage{fix-cm} 

\usepackage{cmap} 
\usepackage[utf8]{inputenc}
\usepackage[english]{babel}
\usepackage[T1]{fontenc}
\usepackage{amsmath}
\usepackage{amsfonts}
\usepackage{bm}
\usepackage{xcolor}
\definecolor{blueviolet}{rgb}{0.2, 0.2, 0.6}
\definecolor{webgreen}{rgb}{0,.5,0}
\definecolor{webbrown}{rgb}{.6,0,0}
\usepackage[pdftex,
	bookmarks=false,
	colorlinks=true, 
	urlcolor=webbrown, 
	linkcolor=blueviolet, 
	citecolor=webgreen,
	pdfstartpage=1,
	pdfstartview={FitH},  
	bookmarksopen=false
	]{hyperref}
\allowdisplaybreaks
\usepackage{tikz}
\usepackage{braket}
\usepackage{natbib}
\usepackage{amsthm}
\usepackage{dsfont}
\usepackage{listings}

\usepackage[T1]{fontenc}
\usepackage{lmodern}
\usepackage{fix-cm}

\numberwithin{equation}{section}

\newcounter{lemc}

\newtheorem{theorem}{Theorem}

\newtheorem{definition}{Definition}

\newtheorem{lemma}[lemc]{Lemma}

\newtheorem{example}{Example}

\theoremstyle{definition}

\newcommand{\Id}{\ensuremath{\mathbb{I}}}

\providecommand{\myvec}[1]{\ensuremath{\boldsymbol{#1}}}

\providecommand{\gg}{\ensuremath{\myvec{g}}}

\providecommand{\ll}{\ensuremath{\myvec{l}}}

\providecommand{\gg}{\ensuremath{\myvec{g}}}

\providecommand{\ll}{\ensuremath{\myvec{l}}}

\def\01{\{0,1\}}

\newcommand{\vertiii}[1]{{\left\vert\kern-0.25ex\left\vert\kern-0.25ex\left\vert #1 \right\vert\kern-0.25ex\right\vert\kern-0.25ex\right\vert}}

\newcommand{\rom}[1]{\mathtt{\uppercase\expandafter{\romannumeral #1\relax}}}

\DeclareMathOperator{\Tr}{tr}
\DeclareMathOperator*{\E}{{\mathbb{E}}}

\usepackage{mdframed}
\usepackage{xcolor}

\definecolor{puzzlebackground}{RGB}{247, 250, 253}    
\definecolor{puzzleborder}{RGB}{70, 142, 175}   
\definecolor{openQbackground}{RGB}{255, 248, 244}  
\definecolor{openQborder}{RGB}{188, 140, 112} 

\newcounter{puzzlecount}
\newenvironment{puzzlebox}[1]
{\begin{mdframed}[backgroundcolor=puzzlebackground,
    linecolor=puzzleborder,
    linewidth=3pt,
    leftmargin=0.1cm,
    rightmargin=0.1cm,
    skipabove=1.2\baselineskip]
    \refstepcounter{puzzlecount}
    {\vspace{0.5em}\large\bfseries Puzzle \thepuzzlecount: #1}\\[0.5em]}
{\end{mdframed}\vspace{1.\baselineskip}}

\newcounter{openQcount}
\newenvironment{openQbox}[1]
{\begin{mdframed}[backgroundcolor=openQbackground,
    linecolor=openQborder,
    linewidth=3pt,
    leftmargin=0.1cm,
    rightmargin=0.1cm,
    skipabove=0.6\baselineskip]
    \refstepcounter{openQcount}
    {\vspace{0.5em}\large\bfseries Open question \theopenQcount: #1}\\[0.5em]}
{\end{mdframed}\vspace{0.8\baselineskip}}

\DeclareMathOperator{\poly}{poly}

\newcommand{\ketbra}[2]{\lvert #1 \rangle \! \langle #2 \rvert}

\begin{document}
\fontsize{10}{12}\selectfont

\title{The vast world of quantum advantage}

\author{Hsin-Yuan Huang}
\affiliation{California Institute of Technology, Pasadena, California 91125, USA}
\affiliation{Google Quantum AI, Venice, California 90291, USA}

\author{Soonwon Choi}
\affiliation{Massachusetts Institute of Technology, Cambridge, Massachusetts 02139, USA}

\author{Jarrod R. McClean}
\affiliation{Google Quantum AI, Venice, California 90291, USA}

\author{John Preskill}
\affiliation{California Institute of Technology, Pasadena, California 91125, USA}
\affiliation{AWS Center for Quantum Computing, Pasadena, California 91125, USA}

\date{\today}

\begin{abstract}
The quest to identify quantum advantages lies at the heart of quantum technology. While quantum devices promise extraordinary capabilities, from exponential computational speedups to unprecedented measurement precision, distinguishing genuine advantages from mere illusions remains a formidable challenge. In this endeavor, quantum theorists are like prophets attempting to foretell the future, yet the boundary between visionary insight and unfounded fantasy is perilously thin. In this perspective, we examine our mathematical tools for navigating the vast world of quantum advantages across computation, learning, sensing, and communication. We explore five keystone properties: predictability, typicality, robustness, verifiability, and usefulness that define an ideal quantum advantage, and envision what new quantum advantages could arise in a future with ubiquitous quantum technology. We prove that some quantum advantages are inherently unpredictable using classical resources alone, suggesting a landscape far richer than what we can currently foresee. While mathematical rigor remains our indispensable guide, the ultimate power of quantum technologies may emerge from advantages we cannot yet conceive.
\end{abstract}

\maketitle

\section{The quest for quantum advantage}

\noindent Quantum entanglement is often portrayed as the epitome of quantum strangeness. Even when two entangled particles are separated by light years, measuring one particle causes the other to instantly ``know'' and collapse to a correlated state. This phenomenon appears to enable faster-than-light communication, violating the fundamental speed limit that governs our classical universe. Could this mysterious quantum effect enable instant galactic communication? The answer, as we now understand, is no. Yet this question reveals a deeper challenge that continues to perplex physicists: When does quantum physics truly outperform classical physics?

This quest to identify genuine quantum advantages, tasks where quantum systems fundamentally surpass what is achievable by their classical counterparts, lies at the heart of quantum technology. While quantum devices are heralded for their revolutionary potential, from quantum computers promising exponential speedups to quantum sensors achieving unprecedented precision, distinguishing genuine quantum advantages from mere illusions has proven remarkably subtle. Entangled particles illustrate this perfectly: despite exhibiting correlations that seem to defy classical physics, quantum entanglement does not necessarily enable capabilities beyond what classical systems can achieve.

\begin{puzzlebox}{Is this spooky? \label{puzzle:spooky}}
Consider two observers, Alice and Bob, sharing an entangled quantum state:
\begin{equation}
    \frac{1}{\sqrt{2}}(\ket{\uparrow} \otimes \ket{\uparrow} + \ket{\downarrow} \otimes \ket{\downarrow}).
\end{equation}
Alice and Bob are separated to opposite ends of the galaxy.
\begin{itemize}
    \item When Alice measures her particle's spin in the $\{\ket{\uparrow}, \ket{\downarrow}\}$ basis, Bob's particle instantly collapses to $\ket{\uparrow}$ or $\ket{\downarrow}$ depending on Alice's outcome.
    \item When Bob measures his spin in the $\{\ket{\uparrow}, \ket{\downarrow}\}$ basis, he immediately knows Alice's result.
    \item This happens instantly and is faster than light could travel between them.
\end{itemize}
Can we simulate this behavior using only classical objects, like a pair of socks? Does the answer change if Alice and Bob randomly choose to measure in either the $\{\ket{\uparrow}, \ket{\downarrow}\}$ basis or the rotated basis $\{\ket{\nearrow}, \ket{\swarrow}\}$? The solution can be found in Appendix~\ref{app:spooky-puzzle}.
\vspace{0.5em}
\end{puzzlebox}

This puzzle illustrates a crucial principle about quantum advantages: they can be remarkably subtle. When Alice and Bob measure only in the $\{\ket{\uparrow}, \ket{\downarrow}\}$ basis, their quantum correlations can be perfectly mimicked by classical systems, such as pairs of socks with predetermined properties. Yet a small change, allowing measurements in multiple bases, reveals a genuine quantum advantage that cannot be replicated by any classical system with local interactions. This fundamental impossibility was established in Bell's elegant theorem~\cite{bell1964einstein}, which demonstrated the non-local nature of quantum correlations. Such delicate distinctions between quantum and classical capabilities pervade the study of quantum advantages.

As quantum technology advances, distinguishing genuine advantages from \emph{pseudo-advantages} has become increasingly crucial. A pseudo-advantage arises when a quantum protocol appears to outperform all classical approaches, only to be matched by a clever classical strategy that was previously overlooked. Puzzle~\ref{puzzle:spooky} demonstrates that while quantum entanglement seems to enable non-local instantaneous correlations, the same statistical outcomes can be achieved classically with pairs of socks, at least when measurements are restricted to a single basis.

The subtlety of quantum advantages extends far beyond this example. Even when quantum algorithms appear to achieve extraordinary feats impossible for classical computers, careful analysis sometimes reveals surprising classical alternatives. The following puzzle illustrates how quantum advantages can be particularly deceptive in the realm of data analysis and machine learning.

\begin{puzzlebox}{Exponential quantum advantage in analyzing big data?}\label{puzzle:beat-qram}
Consider a massive dataset represented by $M$ exponentially long vectors $\vec{x}_1, \ldots, \vec{x}_M \in \mathbb{R}^{2^n}$, each having unit norm. Alice and Bob want to analyze this massive dataset, but they have very different tools at their disposal. Alice can harness the power of quantum technology. Suppose she can create quantum states representing her data in just $\mathrm{poly}(n)$ time that eventually stores the state in $\mathcal{O}(n)$ qubits of memory:
\begin{equation}
    \ket{x_k} = \sum_{i = 1}^{2^n} x_{k,i} \ket{i}.
\end{equation}
Using this ability and a basic quantum algorithm, Alice can efficiently compute inner products between any two vectors using the stored states. Even though the vectors have $2^n$ components, she can estimate:
\begin{equation}
    \braket{x_k | x_{\ell}} = \sum_{i=1}^{2^n} x_{k, i} \cdot x_{\ell, i} = \vec{x_k} \cdot \vec{x_{\ell}}
\end{equation}
to precision $1/\mathrm{poly}(n)$ in $\mathrm{poly}(n)$ time. This seems almost magical as Alice can analyze relationships between vectors of exponential size in time that grows logarithmically in their dimensions. Meanwhile, Bob only has access to classical technology: a standard random-access memory (RAM) with size $\mathrm{poly}(n)$ able to retrieve chunks of the same data at a time. At first glance, his task seems hopeless, as reading all entries of a single vector would take exponential time in $n$.

Can Bob's classical computation somehow match Alice's quantum prowess? Or has Alice discovered a genuine quantum advantage?
The solution can be found in Appendix~\ref{app:beat-qram}.
\vspace{0.5em}
\end{puzzlebox}

At first glance, Alice's quantum approach seems unbeatable. However, Bob has a clever solution using only classical importance sampling to exactly match Alice's quantum performance. The inner product can be estimated in classical polynomial time without any quantum technology.  We also note that Alice's ability to store the information into the amplitudes of a state in $\mathrm{poly}(n)$ is sometimes referred to as quantum random access memory (QRAM). The best known scheme for creating fault-tolerant QRAM generally requires a quantum spacetime footprint scaling as $\mathcal{O}(2^n)$, so this advantage can disappear due to the existing challenge in physical implementations~\cite{jaques2023qram, dalzell2025distillation}.

Such pseudo-advantages have appeared repeatedly throughout the development of quantum technology. Perhaps the most dramatic story began with the HHL algorithm, a quantum algorithm for solving linear systems of equations that exemplifies genuine quantum advantage. In 2009, Harrow, Hassidim, and Lloyd proved that quantum computers can solve linear systems superpolynomially faster than classical computers~\cite{harrow2009quantum}. More precisely, they proved that this superpolynomial quantum advantage follows from the widely believed conjecture that quantum computer is strictly more powerful than classical computers in at least one computational tasks, or $\mathsf{BPP} \neq \mathsf{BQP}$. This breakthrough inspired hundreds of other \emph{derived} quantum algorithms for machine learning, optimization, and numerical linear algebra, many claiming exponential speedups over classical approaches.

However, starting in 2018, this excitement collided with reality. Researchers discovered that some of these derived quantum algorithms could be matched by clever classical approaches \cite{tang2019quantum, tang2021quantum}, similar to Bob's solution in Puzzle~\ref{puzzle:beat-qram}. While the HHL algorithm still has genuine quantum advantage and may provide practical value in certain applications, many previously proposed applications of the HHL algorithm for practically useful tasks turn out to offer only modest computational speedups. This pattern, apparent quantum advantages dissolving under careful analysis, teaches us a profound lesson: even when building upon genuine quantum advantages, we must rigorously verify that the advantage persists in our specific application of interest.

The story of quantum advantages often takes multiple unexpected turns. Consider one of the first enticing facts one learns about in the first chapter quantum computing textbook: a quantum state of $n$ qubits contains $2^n$ complex amplitudes, suggesting an exponential advantage in information storage compared to classical bits. However, the next chapter usually dashes these hopes as Holevo's bound~\cite{holevo1973bounds} proves that only $\mathcal{O}(n)$ classical bits of information can be reliably retrieved from $n$ qubits through measurement. Yet remarkably, this is not the end of the story. The following puzzle illustrates how quantum advantages can reemerge in subtle ways.

\begin{puzzlebox}{Can quantum systems encode exponentially many bits? \label{puzzle:Holevo-vs-Raz}}
Consider a large, normalized vector $\vec{x}$ of dimension $N = 2^n$ stored in Alice's classical memory. Alice wishes to compress this data into $\log N$ qubits to save space and bandwidth when sending it to Bob. She encodes her data into quantum amplitudes using $n = \log N$ qubits:
\begin{equation}
    \ket{x} = \sum_{i = 1}^{2^n} x_{i} \ket{i}
\end{equation}
Satisfied with her extremely compact encoding, Alice realizes she can now transmit or teleport just $n$ qubits instead of the exponentially larger original data. She sends several copies to Bob and keeps a few herself before deleting the original data. Bob stores the quantum states for later processing but wonders, given Holevo's bound, could Alice have just sent $\mathcal{O}(n)$ classical bits and saved all the trouble of using quantum technology.

Did Alice and Bob gain anything by compressing the data into quantum states? Does compressing data into quantum states provide an advantage? Is this compression a genuine quantum advantage, or yet another illusion?
The solution can be found in Appendix~\ref{app:Holevo-vs-Raz}.
\vspace{0.5em}
\end{puzzlebox}

At first glance one may conclude that quantum states offer a simple way to exponentially compress data for communication. However, Holevo's bound appears to quickly and completely shatter this hope as Bob can only extract $\mathcal{O}(n)$ classical bits from his quantum states, suggesting Alice could have achieved the same result by sending classical bits. Yet this is not the complete story. The answer crucially depends on what Bob wants to do with the data. For some tasks, like computing the overlap with another vector, the compression into a quantum state offers no advantage~\cite{tang2019quantum}. However, for other tasks, such as determining whether the vector belongs to a particular subspace or its orthogonal complement, compressing data into a quantum state provides a genuine exponential quantum advantage. Classical communication for solving this problem requires $\mathcal{O}(2^n)$ bits, while quantum communication only requires $\mathcal{O}(n)$ qubits~\cite{raz1999exponential}. Moreover, this quantum advantage does not rely on any unproven computational conjecture. This type of advantage has found recent applications in quantum machine learning, such as for computing quadratic functions of data~\cite{gilboa2024exponential}.

These examples underscore the importance of rigorous analysis of quantum advantages. As quantum technologies rapidly advance and require substantial investments of human effort and financial resources, establishing rigorous theoretical foundations becomes essential to ensure these efforts are well-directed. Quantum computing ventures now attract billions of dollars in investment, and governments worldwide launch major quantum initiatives, making the stakes for establishing genuine quantum advantages increasingly important. Without solid theoretical foundations, we risk pursuing technological directions that may ultimately prove less effective than anticipated, potentially misallocating significant resources and undermining confidence in quantum technologies. The quantum computing community continues to grapple with these challenges, debating evaluation criteria for quantum algorithms and addressing misconceptions about near-term capabilities, while recognizing the critical need for algorithmic breakthroughs to justify ongoing hardware investments~\cite{zimboras2025myths, aaronson2025future, king2025quantum, lanes2025framework}.

At the same time, the quest for genuine quantum advantages offers rewards that extend far beyond responsible stewardship of resources. Seeking quantum advantages requires uncovering new insights into the fundamental distinctions between quantum and classical physics. These insights often transcend their original context, offering new perspectives in fields ranging from condensed matter physics to quantum gravity, from computational complexity to information theory. These profound intellectual rewards provide yet another compelling reason to pursue rigorous quantum advantages. The careful analysis of quantum advantages thus serves both the responsible stewardship of society's investment and our fundamental understanding of the quantum universe.

\begin{figure}
    \centering
    \includegraphics[width=1.0\linewidth]{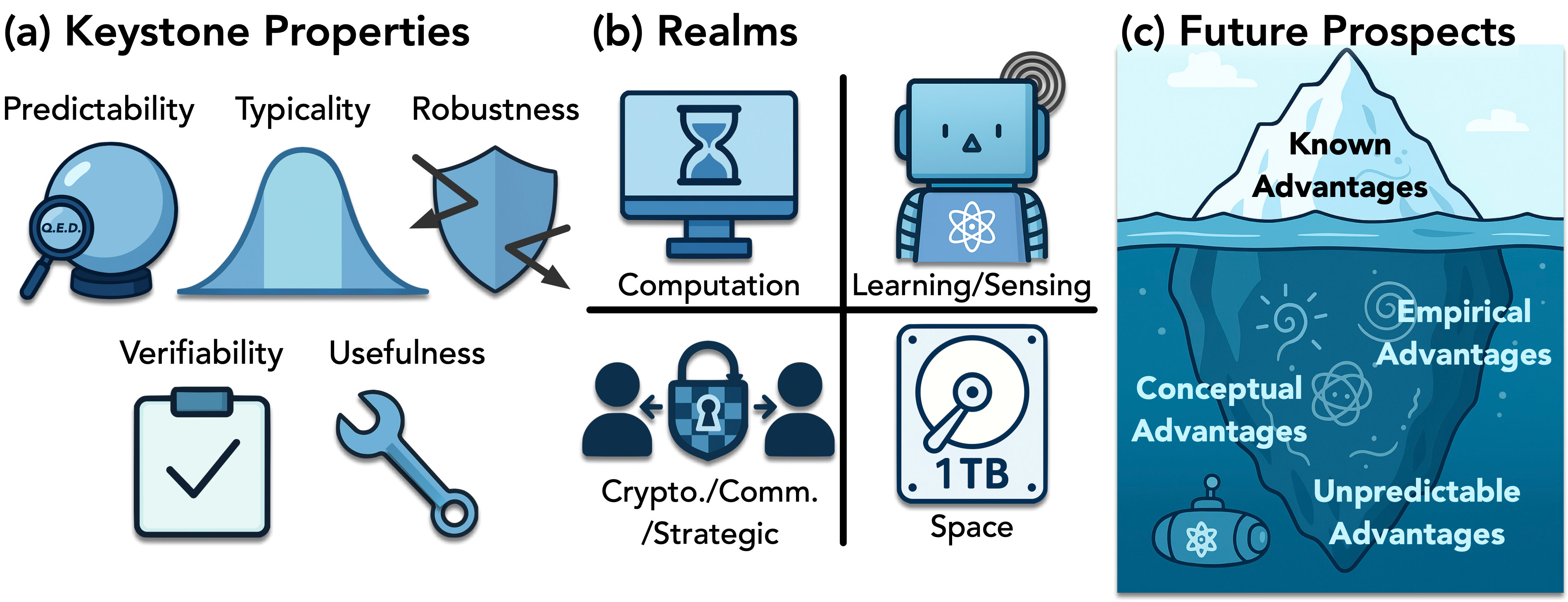}
    \caption{
        \textsc{Conceptual overview of the vast world of quantum advantage.}
        (a) \textsc{Keystones of quantum advantage.} We propose five essential properties that collectively define an ideal quantum advantage: \emph{Predictability} (supported by rigorous evidence), \emph{Typicality} (applying to a substantial fraction of naturally occurring problem instances), \emph{Robustness} (maintaining advantage despite noise and imperfections), \emph{Verifiability} (enabling efficient validation of correctness), and \emph{Usefulness} (providing genuine practical value).
        (b) \textsc{Realms of quantum advantage.} Quantum advantages can be found across four primary domains: \emph{Computation} (algorithmic speedups), \emph{Learning/Sensing} (enhanced methods for probing the world), \emph{Cryptographic/Communication/Strategic Games} (advantages in security and coordination), and \emph{Space} (improved memory efficiency).
        (c) \textsc{Future of quantum advantage.} The visible tip of the iceberg represents the currently known and predicted advantages, while the vast, submerged part represents the largely unexplored frontier of empirical, conceptual, and fundamentally unpredictable advantages that may only be discovered through future quantum technologies.
    }
    \label{fig:vast-world}
\end{figure}

\section{Keystones of quantum advantage}

\noindent Having explored why identifying genuine quantum advantages is crucial for the future of quantum technology, we now turn to a systematic analysis of what constitutes compelling examples of quantum advantages. What properties must they possess to be both theoretically convincing and practically valuable? This question is particularly pressing in the current era, where quantum technologies are rapidly advancing yet remain imperfect, and where substantial resources are being allocated to quantum research and development. We propose five keystone properties that define an ideal quantum advantage.

\subsection{Predictability}

All theorists studying quantum technologies are, in essence, prophets trying to predict the future. In the current era, where large-scale quantum technology remains elusive, the exploration of quantum advantages requires careful theoretical guidance to distinguish genuine advantages from pseudo-advantages that may vanish upon closer examination. This motivates our first keystone property: quantum advantages should be predictable, ideally through rigorous, substantive, and quantifiable evidence.

\begin{definition}[Predictability]
A quantum advantage is predictable if we have sufficient evidence supporting that, given the necessary quantum technology, we will achieve capabilities fundamentally beyond what is possible with classical technology.
\end{definition}

The importance of predictability stems from the inherent difficulty of reasoning about quantum advantages. Establishing genuine quantum advantages is a subtle endeavor fraught with hidden pitfalls. Our classical intuitions often fail us when reasoning about quantum systems, and seemingly obvious quantum advantages can dissolve under careful analysis. Thus, we need robust theoretical frameworks not only to distinguish genuine quantum advantages from pseudo-advantages that merely reflect our incomplete understanding, but also to quantify the magnitude of these advantages with mathematical rigor.

\begin{example}[Quantum recommendation systems]
The evolution of quantum recommendation systems illustrates the subtle challenges in predicting quantum advantages. Early work \cite{kerenidis2017quantum} proposed quantum algorithms that appeared to offer exponential speedups for recommendation problems, representing a natural application area where quantum advantages seemed plausible. However, the wide array of classical algorithmic techniques made it difficult to definitively establish the advantage. Building on techniques similar to those in Puzzle~\ref{puzzle:beat-qram}, researchers later found that classical algorithms could achieve comparable performance~\cite{tang2019quantum}, demonstrating how some apparent quantum advantages can dissolve when the full classical algorithmic landscape is thoroughly explored.
\end{example}

This and similar experiences have driven the quantum computing community to develop more rigorous methodologies for establishing predictable and quantifiable advantages. Several powerful paradigms have emerged, each with different strengths for providing reliable predictions about quantum capabilities. The gold standard is formal mathematical proof for separating the optimal performance achievable by classical technology from the performance using quantum technology. Such formal proofs allow us to predict quantum advantages with absolute certainty and precisely quantify their magnitude. Nonlocal games built upon Bell's inequality (explored in Puzzle~\ref{puzzle:spooky}) exemplify this approach.

A second powerful approach proceeds via reductions to widely believed conjectures. The key is identifying conjectures supported by diverse evidence across different fields. We then prove that the absence of quantum advantage would violate these conjectures. Common examples include $\mathsf{BPP} \neq \mathsf{BQP}$ (not all quantum polynomial-time problems are classically solvable in polynomial time), the non-collapse of the polynomial hierarchy \cite{stockmeyer1976polynomial}, and the classical hardness of factoring \cite{shor1997polynomial}. Even quantum computational hardness assumptions, such as the hardness for quantum computers to solve the learning with errors problem \cite{regev2009lattices}, can sometimes be transformed into proofs of quantum advantage \cite{mahadev2018classical, schuster2024random}.

\begin{example}[Shor's algorithm for factoring \cite{shor1997polynomial}]
Shor's algorithm provides compelling evidence for quantum advantage through reduction to the hardness of factoring, a problem underpinning trillions of dollars of global digital commerce \cite{unctad2024digital, forrester2024global}. RSA encryption, standardized by NIST as a federal cryptographic standard \cite{nist2013fips186}, forms a critical component of this digital infrastructure. Despite enormous financial incentives and decades of intensive research by cryptographers and computer scientists, no efficient classical algorithm has emerged. The proven ability of Shor's algorithm to factor efficiently, combined with the sustained difficulty of classical approaches, provides strong evidence for quantum advantage.
\end{example}

\begin{example}[Cooling physical systems \cite{chen2024local}]
A recent result shows how quantum advantage emerges in the seemingly simple task of cooling quantum systems. When a quantum system is placed in a cold environment, it naturally finds a local minimum of its energy. Local minima are not necessarily the absolute lowest (ground) state, but a state that cannot be improved by small changes. While finding ground states is believed hard even for quantum computers, cooling to local minima exhibits a clear quantum advantage: it is efficiently achievable by quantum computers but provably hard for classical ones under the widely-believed conjecture that $\mathsf{BPP} \neq \mathsf{BQP}$. This result is particularly compelling because it connects to a ubiquitous physical process, the natural cooling of quantum systems, while providing rigorous evidence for quantum advantage through reduction to complexity assumptions.
\end{example}

\begin{openQbox}{New Conjectures for Building Quantum Advantages}
Currently, provable quantum advantages rely on relatively few complexity-theoretic conjectures. Can we identify new mathematical conjectures, perhaps inspired by physical principles, that enable us to predict and develop new types of quantum advantages?
\vspace{0.5em}
\end{openQbox}

An emerging third paradigm elegantly combines empirical numerical calculations with mathematical proofs to identify \emph{instance-specific quantum advantages}. For certain quantum algorithms, we can predict their performance on specific problem instances by computing particular quantities that are feasible using classical computers. While quantum technology remains essential for actually solving the original problem, this approach enables rigorous comparison between predicted quantum performance and the capabilities of the best known classical algorithms, allowing us to forecast and quantify quantum advantages before constructing the necessary quantum hardware. Unlike theoretical worst-case analysis, this paradigm enables evaluation on specific problem instances that arise in real applications. While this approach cannot rule out future classical algorithmic breakthroughs, it substantially expands our ability to analyze potential quantum advantages in practically relevant problem instances.

\begin{example}[Quantum approximate optimization algorithm \cite{farhi2014quantum}]
Recent advances in understanding the Quantum Approximate Optimization Algorithm (QAOA)~\cite{farhi2014quantum} illustrate this paradigm. For specific problem classes, QAOA's performance can be predicted through classical numerical computation~\cite{basso2021quantum, basso2022performance, farhi2025lower}. Remarkably, rigorous lower bounds can be established even when running the quantum algorithm is infeasible with current technology. For MaxCut on high-girth 3-regular graphs, this analysis has provided evidence for potential quantum advantages over known classical algorithms~\cite{farhi2025lower}.
\end{example}

\begin{example}[Decoded quantum interferometry \cite{jordan2024optimization}]
Decoded quantum interferometry (DQI)~\cite{jordan2024optimization} provides another compelling illustration. DQI reduces classical optimization problems to classical decoding problems through quantum interference. DQI's performance can be predicted by evaluating classical decoders, even for large-scale instances. This enables comparison with classical algorithms before building quantum hardware. For certain problems, such as optimal polynomial intersection, this approach has revealed that DQI achieves better approximation ratios than any known classical algorithm~\cite{jordan2024optimization}.
\end{example}

\begin{example}[Quantum phase estimation \cite{kitaev1995quantum, lin2020near, wu2024variational}]
This paradigm also manifests in quantum phase estimation (QPE) for finding low-energy states of a Hamiltonian $H$. After classical algorithms like DMRG \cite{white1992density} find their best approximation to the ground state, we can predict how quantum phase estimation would improve upon this result. If the energy variance $\Tr(H^2\rho) - \Tr(H\rho)^2$ of the classical state $\rho$ is substantial, QPE can efficiently find lower-energy states \cite{wu2024variational}. This variance, often efficiently computable classically, provides a rigorous lower bound on the quantum advantage without requiring quantum hardware.
\end{example}

\begin{openQbox}{New Theory for Predicting Quantum Advantages}
How can we expand our mathematical arsenal for predicting and quantifying quantum advantages before having full-fledged quantum technology? Can we develop a systematic theory for forecasting when and how quantum systems will surpass classical limitations, along with precise bounds on the magnitude of these advantages? The ability to rigorously predict and quantify quantum advantages may be as transformative as the advantages themselves.
\vspace{0.5em}
\end{openQbox}

\vspace{-1.2em}
\subsection{Typicality}

When researchers establish quantum advantages for computational, communication, learning, or sensing problems, they often prove that quantum technology can efficiently solve all instances within a problem class, while classical technology provably cannot. Though such results seem ideal, the reality is more nuanced. The distinction between worst-case and average-case complexity plays a crucial role in determining the practical relevance of quantum advantages.

\begin{example}[Simulating quantum circuits] \label{ex:sim-qcirc}
Consider the task of simulating a closed quantum system evolving under a time-dependent Hamiltonian described by a polynomial-size quantum circuit $C$. The goal is to estimate local observables on the output state after evolving a simple product input state under circuit~$C$. Under the widely believed conjecture that $\mathsf{BPP} \neq \mathsf{BQP}$, this task exhibits a superpolynomial quantum advantage in the worst case; that is, for the quantum circuits that are most difficult to simulate classically. However, this quantum advantage only applies for some quantum circuits. For most quantum circuits consisting of independent and generic circuit layers, the quantum chaos in the dynamics enables efficient classical simulation through low-weight Pauli propagation \cite{aharonov2023polynomial, shao2023simulating, nemkov2023fourier, beguvsic2023fast, beguvsic2023simulating, fontana2023classical, rudolph2023classical, angrisani2024classically, dowling2024magic}. Moreover, it is widely believed that under natural Hamiltonian dynamics, local observables approach the predictions from generalized thermal equilibrium states characterized by a few conserved quantities such as particle numbers. This stark contrast between worst-case and average-case complexity exemplifies why worst-case advantages need not imply quantum advantages for typical instances.
\end{example}

The relationship between worst-case and average-case quantum advantages can be quite subtle. A problem may show no quantum advantage in the worst case, yet possess significant advantages for typical instances. Conversely, as illustrated in Example~\ref{ex:sim-qcirc}, a problem may demonstrate dramatic worst-case quantum advantage while admitting efficient classical solutions for most instances.

From a practical perspective, the ability to solve contrived worst-case instances is far less valuable than solving the typical instances encountered in real-world applications. This observation motivates our second criterion for meaningful quantum advantages:

\begin{definition}[Typicality]
A quantum advantage exhibits typicality if it applies to a substantial fraction of practically relevant problem instances, rather than being confined to particular instances that are especially hard classically.
\end{definition}

This definition emphasizes that quantum advantages should apply broadly to the kinds of problems we actually want to solve, not just carefully crafted instances to prove quantum-classical separations. A typical quantum advantage promises practical impact across a broad spectrum of applications, making it far more valuable than a narrow advantage limited to pathological cases that rarely occur in practice.

To establish typical quantum advantages, we can build on average-case complexity theory, a subfield of theoretical computer science that provides tools for analyzing algorithmic performance on typical problem instances. Unlike worst-case complexity, which considers the maximum computational resources required across all possible inputs, average-case complexity examines the expected resource consumption under natural probability distributions of problem instances. While this theoretical framework is still developing, it has already yielded several compelling examples of typical quantum advantages.

\begin{example}[Random circuit sampling \cite{arute2019quantum, zhu2022quantum, morvan2023phase, abanin2025constructive, gao2025establishing}]
Random circuit sampling (RCS) exemplifies how average-case hardness can establish typical quantum advantages. The task involves sampling from the output distribution of a random quantum circuit to within a specified total variation distance. For example, Google's Willow processor demonstrated quantum supremacy by sampling from a $103$-qubit random circuit in just a few minutes, while a classical computer would require much longer than the age of the universe to complete the same task \cite{abanin2025constructive}. The classical hardness of RCS stems from a reduction. For most circuit instances, approximate sampling to small total variation distance requires superpolynomial classical computation, under the widely believed assumption that the polynomial hierarchy does not collapse. These results establish that quantum advantages for RCS are typical. They hold for most random quantum circuits rather than carefully constructed worst-case instances.
\end{example}

\begin{example}[Discrete logarithm problem]
The discrete logarithm problem (DLP) provides another prototypical example of typical quantum advantage. Given a cyclic group $G$ of prime order $p$ with generator $g$, and an element $h \in G$, the task is to find $x$ such that $g^x = h$. While Shor's algorithm solves DLP in quantum polynomial time, the best-known classical algorithms require exponential time for most group elements when $G$ is carefully chosen \cite{boneh1998decision}. The generic nature of this quantum advantage stems from a powerful reduction: finding discrete logarithms for a random element $h \in G$ is provably as hard as solving the problem in the worst case \cite{shoup1997lower}. This reduction is known as random self-reducibility and holds for broad classes of groups, including those used in practical cryptographic systems, making the average-case hardness a foundation for widely deployed cryptographic protocols.
\end{example}

These examples highlight how average-case complexity can establish quantum advantages that are both theoretically rigorous and practically relevant. However, developing the theory of average-case complexity remains an ongoing challenge. The field must address several fundamental questions: identifying natural probability distributions over problem instances that capture real-world applications, finding more problems with properties like random self-reducibility that connect average-case to worst-case hardness, and developing new proof techniques for average-case lower bounds on classical computation. As of now, most superpolynomial quantum advantages are not known to be typical. The most well known examples are random circuit sampling, where the randomly generated bitstrings are not known to be of much use, and Shor's algorithm for breaking cryptographic protocols, where one can reinstate security by using post-quantum cryptography \cite{regev2009lattices}.  More recently, the problem of optimal polynomial interpolation (OPI) has been conjectured to be another such instance. For OPI, there are parameter regimes that are efficiently solvable by DQI~\cite{jordan2024optimization} but which are not solved in polynomial time by any known classical algorithms, even for average-case instances, though this problem is not known to have the property of random self-reducibility like DLP. Aside from such superpolynomial speedups, there are also a variety of average-case-hard problems, for example Circuit-SAT, that admit Grover-like quadratic speedups for typical instances.  More recently, quantum algorithms have been discovered giving a super-quadratic polynomial speedup for the problem of sparse learning parity with noise (LPN)~\cite{schmidhuber2025quartic}. As LPN is a canonical average-case-hard problem where even a quantum computer is not expected to offer an efficient solution in general, better-than-quadratic polynomial speedups for typical instances of this problem are conceptually significant.  These observations raise the following open question.

\vspace{0.6em}
\begin{openQbox}{Typical quantum advantage}
Are there other problems not based on random circuit sampling, Shor's algorithm, or OPI that exhibit typical, superpolynomial quantum advantage, where the quantum advantage applies to most instances of the problem rather than just worst-case instances? Can we establish new tools enabling us to show that broad families of problems have typical quantum advantage?
\vspace{0.5em}
\end{openQbox}

\vspace{-1.2em}
\subsection{Robustness}
\label{subsec:robustness}

\noindent The journey from theoretical to empirical quantum advantage often faces a formidable enemy: noise. Quantum systems are notoriously fragile as interactions with the environment easily destroy quantumness. Moreover, quantum advantages may erode under deviations from idealized theoretical assumptions. This raises a crucial question: can quantum advantages survive under noise and imperfection?

\begin{definition}[Robustness]
A quantum advantage is robust if it persists despite deviations from idealized theoretical assumptions, including imperfections in hardware, available resources, input data, etc.
\end{definition}

For quantum computing tasks, we possess the remarkable theoretical framework of quantum fault tolerance that provides a path to universal hardware robustness. The threshold theorem proves that if we can reduce errors in our quantum hardware below a certain threshold, we can achieve arbitrarily long quantum computations at reasonable overhead using carefully designed error correction protocols. This means quantum computational advantages, like Shor's factoring algorithm or quantum simulation of many-body systems, are theoretically robust: they maintain their advantage even with noisy quantum hardware, provided the noise is below threshold.

\begin{example}[Fault-tolerant quantum computation]
The threshold theorem ensures that a quantum circuit with $C$ perfect gates can be simulated by a noisy quantum circuit with error at most $\epsilon$, provided the noise per gate is below a threshold $p_{\text{th}}$. Traditional approaches require $\mathcal{O}(C \log^c(C/\epsilon))$ noisy gates for this simulation, where $c$ is a constant. Recent breakthroughs have dramatically improved this overhead: using constant-rate quantum locally testable codes, a depth-$D$ quantum circuit on $n$ qubits with $C$ ideal gates can be simulated by a noisy quantum circuit that maintains a constant space overhead (using only $\mathcal{O}(n)$ qubits) and near-logarithmic time overhead (using only $D \log^{1+o(1)}(C/\epsilon)$ depth)~\cite{gottesman2013fault,nguyen2024quantum}.
\end{example}

However, the story changes dramatically when we consider other kinds of quantum advantages. Quantum sensing offers a striking example.

\begin{example}[Entanglement-enhanced sensitivity]\label{ex:sensing}
By exploiting quantum entanglement, quantum sensors can achieve precision scaling as $1/(NT)$ using an $N$-particle probe with total sensing time $T$; this so-called Heisenberg-limited quantum enhanced scaling should be compared to the standard $1/(\sqrt{N} T)$ scaling achieved by an unentangled probe. However, this advantage dissolves in the presence of generic noise. Recent no-go theorems~\cite{zhou2018achieving} prove that, for many sensing tasks, even small amounts of noise can completely destroy the asymptotic quantum advantage, forcing the precision to the scaling $1 / \sqrt{NT}$ for both entangled and unentangled sensors. The $\sqrt{N}$ scaling advantage achieved using entangled probes in the ideal case is fundamentally unachievable when generic noise is present, even if quantum error correction is strategically deployed. A simplified, self-contained proof is given in Appendix~\ref{app:no_heisenberg}.
\end{example}

\noindent This example presents a stark contrast to the case of computation, where the exponential advantage of a quantum computer over a classical computer survives even when the quantum computer is subject to generic noise provided the noise is sufficiently weak.

Beyond hardware noise, quantum advantages must also withstand deviations from other theoretical assumptions. The problem instances occurring in the real world may not satisfy properties that guarantee quantum advantages.  In physical simulation, the initial condition or the specification of the dynamics may be slightly different from those considered theoretically. And in machine learning, the classical datasets may be inherently noisy or incomplete \cite{fischler1981random, tibshirani1996regression, natarajan2013learning}.

In recent years, theorists have envisioned quantum artificial intelligence (AI) agents that can learn about our universe by retrieving information using quantum sensors, storing the information in quantum memory, and processing the information using quantum computers. Ideal noiseless quantum AI agents have been proven to achieve some learning tasks with exponentially fewer experiments than would be required by classical AI agents that store data in classical memory and process it with classical computers \cite{huang2021information, chen2021exponential, huang2021quantum}. Because quantum AI agents integrate multiple quantum technologies, they face robustness challenges across multiple fronts simultaneously. The vulnerability of quantum sensing to noise may hamper data collection, while memory decoherence can corrupt stored information, and computational errors can compromise data processing. While some quantum learning advantages have been proven to persist under realistic noise conditions~\cite{huang2021quantum, huang2022foundations, chen2023complexity} and experimentally demonstrated in intermediate-scale superconducting qubit platforms~\cite{huang2021quantum} and photonic systems~\cite{liu2025quantum}, the robustness of many other learning advantages remains uncharacterized. In particular, it is not known whether universal fault-tolerant schemes exist that can preserve all exponential advantages enjoyed by ideal quantum AI agents.

\vspace{0.5em}
\begin{openQbox}{Fault-tolerant quantum AI agents}
Are there universal fault-tolerant schemes that enable a noisy quantum AI agent to efficiently simulate any noiseless ideal quantum AI agent? Such schemes must preserve the agent's learning advantages while maintaining robustness against multiple simultaneous sources of imperfection across sensing, memory, and computation.
\vspace{0.5em}
\end{openQbox}

\vspace{-1.2em}
\subsection{Verifiability}

\noindent Quantum technology operates at the frontier of physics, where highly complex entanglement and superposition dominate. As we push towards large-scale quantum systems, we venture into regimes of complexity and entanglement never before encountered in physical experiments. This raises a profound question: can we trust quantum mechanics to remain valid in these unexplored regimes?

The history of physics suggests caution. Just as Newton's laws break down at high velocities or in strong gravitational fields, superseded by Einstein's special or general relativity, quantum mechanics might behave differently in regimes of high complexity and large-scale entanglement. Our confidence in quantum theory comes primarily from experiments in relatively simple quantum systems with limited entanglement. We have never tested quantum mechanics at the complexity frontier where thousands of particles maintain complex quantum correlations. As we push more deeply into this complexity frontier, how can we trust that our quantum devices are actually achieving what we think they are?

This concern isn't merely theoretical. The fragile nature of quantum systems means that errors can accumulate and propagate in subtle ways, potentially leading to incorrect outputs that might nevertheless appear plausible. Moreover, as we venture into regimes of high entanglement that have never been experimentally explored, our intuition and standard verification techniques may fail us. This challenge motivates our fourth criterion for meaningful quantum advantages:

\begin{definition}[Verifiability]
A quantum advantage is verifiable if it can be checked efficiently that the quantum technology is producing the correct output.
\end{definition}

Verification serves a dual purpose in quantum technology. First, it ensures our quantum devices are functioning as intended. For some quantum advantages, this verification is straightforward. In quantum communication protocols demonstrating Bell inequality violations, the correlations between measurement outcomes directly verify the quantum advantage. Similarly, for quantum algorithms like Shor's factoring, the output can be efficiently verified by classical multiplication.

Second, and perhaps more profoundly, verification protocols can serve as stringent tests of quantum mechanics itself in previously unexplored regimes. When a highly complex quantum system behaves exactly as quantum theory predicts, especially in scenarios where classical simulation is provably impossible, we gain confidence in the theory's validity at higher scales of complexity. Conversely, any deviation from quantum predictions in these regimes could signal the discovery of new physics. However, verification becomes considerably more challenging for tasks involving complex quantum states or dynamics.

\begin{example}[Classical verification of quantum computation \cite{mahadev2018classical}]
A breakthrough in verifiable quantum advantage came with protocols that enable a classical computer to verify any quantum computation. These protocols allow a classical client to delegate quantum computations to a quantum server while maintaining both privacy and verifiability. The verification relies on cryptographic assumptions, such as the quantum hardness of the learning with errors problem. Through an elegant combination of cryptographic primitives and measurement-based quantum computation, these protocols ensure that any deviation from the prescribed quantum computation is detected with high probability. Most remarkably, the verification can be performed entirely by classical computers, providing a bridge between classical and quantum computation that enables rigorous testing of quantum advantages. This result demonstrates that a classical system can verify a quantum computation even when it cannot efficiently simulate the quantum computation.
\end{example}

While universal verification schemes exist for quantum computation, the landscape of verifiable quantum advantages in other domains remains largely unexplored. In quantum sensing, verification often relies on indirect evidence or statistical consistency checks. For quantum machine learning, verifying that a quantum model has truly learned better features or representations than alternative classical models remains challenging. This leads to several open challenges:

\vspace{0.5em}
\begin{openQbox}{Verification beyond computation}
Can we develop general verification schemes for quantum advantages in communication, learning, and sensing that go beyond computation? How can we verify outputs from quantum technology when we venture into new physical regimes?
\vspace{0.5em}
\end{openQbox}

\vspace{-1.2em}
\subsection{Usefulness}

\noindent The ultimate test of any technological advantage lies not in its theoretical power but in its practical impact. While demonstrating a quantum advantage that is predictable, generic, robust, and verifiable is itself a remarkable scientific achievement, these properties alone don't guarantee that the advantage will be valuable in real-world applications. This observation leads to our final criterion:

\begin{definition}[Usefulness]
A quantum advantage is useful if the advantage provides practical value to a user who seeks an effective method for a significant application (and who does not care whether the underlying technology is classical or quantum).
\end{definition}

The emphasis on practical value for users who don't care about the underlying physics is crucial. Many early quantum algorithms focused on problems specifically constructed to showcase quantum advantages, like random circuit sampling \cite{arute2019quantum} or boson sampling \cite{aaronson2011computational}. While these demonstrate fascinating quantum-classical separations, their practical utility remains limited. In contrast, truly useful quantum advantages should address problems we genuinely want to solve, regardless of how we solve them.

\begin{example}[Quantum simulation of materials]
Quantum simulation of strongly correlated quantum materials exemplifies a potentially useful quantum advantage. Understanding and predicting the properties of complex materials like high-temperature superconductors, exotic quantum magnets, or catalyst molecules has immense scientific and technological value, independent of the methods used. Classical computers face fundamental challenges in simulating these systems due to the intricate quantum correlations present in many-body systems, while quantum computers can naturally represent and evolve these quantum states. However, the landscape of simulation methods continues to evolve rapidly. Classical algorithms, particularly those leveraging machine learning and tensor networks \cite{orus2014practical}, have shown remarkable progress in describing previously intractable quantum systems. It remains challenging to predict whether future classical technological breakthroughs will enable efficient simulation of the quantum materials most relevant to practical applications.
\end{example}

The path to useful quantum advantages often involves identifying problems where quantum effects are inherent to the challenge at hand, rather than forcing quantum approaches onto classical problems. Moreover, usefulness must be evaluated in the context of rapidly evolving classical capabilities and the practical constraints of quantum technology. It might turn out that the most impactful quantum advantages are not those with the most impressive theoretical speedups, but instead the applications where quantum approaches offer qualitatively new capabilities or insights.

Finding new quantum advantage that achieves all five keystone properties simultaneously remains a grand challenge for the field. This is stated in the following open question:

\vspace{0.6em}
\begin{openQbox}{Ideal quantum advantage}
What quantum advantages are simultaneously predictable, typical, robust, verifiable, and useful? Can we develop new approaches to discover quantum advantages that satisfy all these criteria?
\vspace{0.4em}
\end{openQbox}

\section{Realms of quantum advantage}

The vast world of quantum advantage extends far beyond computational speedups and presents a rich tapestry of opportunities where quantum resources enable capabilities beyond those achievable with classical resources alone. We can organize these quantum advantages into distinct realms, each characterized by fundamentally different sources of quantum superiority over classical approaches. These realms differ not only in their applications but in their fundamental nature. Some rely on unproven complexity assumptions, while others are guaranteed by the laws of physics themselves. For comprehensive surveys of quantum algorithms and protocols across these realms, we refer readers to existing resources~\cite{degen2017quantum, pirandola2020advances, dalzell2023quantum, QuantumAlgorithmZoo}.

\subsection{Computational advantages}

The computational realm constitutes the most extensively studied realm of quantum advantages, characterized by algorithms that promise faster solutions to computational problems. Yet this realm has a complex landscape where genuine advantages coexist with subtle illusions, and where the foundation of each advantage rests on computational complexity assumptions.

\vspace{0.5em}\paragraph*{Complexity-theoretic foundation.} Unlike other realms we will explore, computational quantum advantages fundamentally depend on unproven conjectures about the computational difficulty of certain problems. The canonical assumption $\mathsf{BPP} \neq \mathsf{BQP}$ asserts that quantum computers can efficiently solve some problems that are intractable for classical computers. While this assumption is supported by substantial evidence, it remains a conjecture, one that underpins virtually every claimed computational quantum advantage.

\vspace{0.5em}\paragraph*{Oracle-based advantages.} The theoretical power of quantum computation was first established in oracular settings \cite{bernstein1993quantum, simon1997power}, where an oracle for computing a function is given as a black box. Simon's algorithm~\cite{simon1997power} provided the first exponential advantage and directly inspired the development of Shor's algorithm. Other foundational results include Grover's search \cite{grover1996fast} and quantum walks \cite{farhi1998quantum, childs2003exponential, farhi2007quantum} exhibiting polynomial to exponential oracular advantages. While oracle-based advantages provide important conceptual foundations for understanding quantum computational power, they are not computational advantages in the usual sense. The black box assumption means that when the oracle's internal structure is revealed, the quantum advantage could disappear as classical algorithms may find efficient shortcuts that exploit the revealed structure.

\vspace{0.5em}\paragraph*{Shor's algorithm.} Building on the theoretical insights from Simon's algorithm, Shor's algorithm~\cite{shor1997polynomial} exemplifies genuine computational quantum advantage. It exploits the hidden subgroup structure that quantum computers can efficiently solve, providing a concrete path from oracle-based advantage to real-world impact. As discussed in earlier sections, it satisfies our keystone criteria and demonstrates how quantum advantage can have profound real-world implications through its threat to cryptographic infrastructure worldwide. Even decades after its discovery, no efficient classical algorithm has emerged to challenge its speedup.

\vspace{0.5em}\paragraph*{Quantum simulation.} Beyond number-theoretic problems, quantum simulation represents a particularly natural computational application, as the task of simulating quantum many-body systems inherently involves quantum mechanical effects that classical computers struggle to capture efficiently~\cite{feynman1982simulating, lloyd1996universal}. The intuition is compelling: quantum systems should naturally simulate quantum physics. Yet this intuitive advantage faces relentless competition from advancing classical methods. Sophisticated tensor network algorithms~\cite{orus2014practical}, machine learning approaches~\cite{carleo2017solving}, and quantum-inspired classical methods~\cite{lee2023evaluating} continue to push the boundaries of classical simulation. While quantum simulators have the potential to achieve advantages for specific problems, the landscape remains dynamic. What appears quantum-advantageous today may be challenged by classical innovations tomorrow.

\vspace{0.5em}\paragraph*{Power of data.} A subtle challenge to computational quantum advantages emerges with the recognition of the power of data~\cite{huang2021power}. When quantum experiments generate data that cannot be efficiently produced classically, even classical machine learning algorithms trained on such data can solve computational problems beyond the reach of classical algorithms operating in isolation~\cite{huang2021provably, lewis2024improved}. This reveals a deeper vulnerability: as we accumulate more data about our quantum universe through quantum simulations, quantum sensors, and quantum experiments, classical machine learning algorithms, by exploiting that data, might learn to approximate quantum phenomena without requiring additional quantum computation. However, assuming $\mathsf{BPP} \neq \mathsf{BQP}$, classical algorithms cannot universally simulate all quantum experiments. This suggests that important applications can be realized by quantum computing and classical AI working together; for example, quantum computers might collect data in previously unexplored physical regimes, empowering classical AI to predict properties of exotic quantum systems that have not yet been observed.

\vspace{1em}
The realm of computational advantages presents a nuanced picture: while the quantum advantage exhibited by foundational algorithms like Shor's remains well grounded, the broader landscape is characterized by ongoing competition between quantum and classical computers. Quantum computational advantages require constant vigilance as the classical state-of-the-art continues to advance.

\subsection{Learning/sensing advantages}

Moving beyond pure computation, the learning and sensing realm offers quantum advantages that emerge when quantum technology interfaces directly with the physical world. Here, advantages arise not from abstract algorithmic improvements, but from the fundamental quantum nature of information gathering, the physical systems we seek to understand, and the ability to process quantum information in ways that classical systems cannot. Crucially, many advantages in this realm can be established without relying on computational complexity assumptions.

\vspace{0.5em}\paragraph*{Physics-based foundation.} Unlike computational advantages that rely on unproven complexity-theoretic conjectures, nearly all known learning and sensing advantages rest on fundamental physical principles. The standard quantum limit in sensing, Heisenberg's uncertainty principle, and the inevitable disturbance of quantum states upon measurement all derive from the rigorous mathematical foundations of quantum mechanics, foundations that have been extensively validated over the past century. These laws of nature enable the development of learning and sensing advantages without requiring any assumptions about computational hardness. This physical grounding endows many learning and sensing advantages with a intrinsic robustness that computational advantages inherently lack.

\vspace{0.5em}\paragraph*{Quantum sensing fundamentals and limits.} Quantum sensing harnesses intrinsically quantum objects, such as atoms, ions, molecules, or photons, to measure physical quantities with enhanced performance. This quantum enhancement requires sophisticated quantum states, exemplified by GHZ states for magnetic field sensing~\cite{giovannetti2004quantum, giovannetti2006quantum}. As described in Section~\ref{subsec:robustness}, achieving this quantum advantage faces a fundamental adversary: noise. The Hamiltonian-not-in-Lindbladian-Span (HNLS) criterion reveals that a quantum enhancement survives in the presence of noise only when the target signal and noise processes can be distinguished at the operator level~\cite{zhou2018achieving}. In generic scenarios, this condition fails, rendering the asymptotic quantum sensing advantage fundamentally unattainable even with quantum error correction~\cite{demkowicz2012elusive}. This stands in stark contrast to quantum computation, where universal fault tolerance provides a clear path to noise robustness.

\vspace{0.5em}\paragraph*{Practical quantum sensing applications.} Despite theoretical limitations, quantum sensing has already achieved remarkable real-world successes. While finding examples where entanglement enhances sensing capabilities in a real-world setting has been famously elusive, a notable exception is LIGO's detection of gravitational waves, employing squeezed light to reduce quantum noise below the standard quantum limit~\cite{ligo2011gravitational,abbott2016observation,schnabel2017squeezed,ganapathy2023broadband}. In contrast, advantages exploiting the innate sensitivity of quantum sensors or ensembles thereof have been far more common. For example, nitrogen-vacancy centers in diamond enable quantum-enhanced magnetometry with nanoscale spatial resolution~\cite{taylor2008high, lovchinsky2016nuclear} or under extreme temperature or pressure conditions~\cite{bhattacharyya2024imaging, wang2024imaging}, while quantum-enhanced atomic clocks push temporal precision to unprecedented levels~\cite{ludlow2015optical, bothwell2022resolving, nicholson2015systematic, schine2022long}. In sensing applications, even constant factor improvements over existing technology can provide transformative practical utility. Unlike computational tasks, where performance can often be enhanced by distributing workloads across multiple processors, sensing tasks frequently face fundamental constraints imposed by their physical infrastructure. For instance, achieving nanoscale spatial resolution in magnetometry requires the sensor to be physically positioned at the target location, while gravitational wave detection demands interferometers with enormous 4-kilometer-long arms. These physical constraints make it prohibitively expensive or physically impossible to simply deploy additional sensors to enhance performance. Quantum enhancements thus offer a fundamentally different value proposition than computational speedups.

\vspace{0.5em}\paragraph*{Quantum learning agents.} Recent developments have found that the most significant advantages emerge when we consider quantum machines that can see with quantum sensors, think with quantum computers, and remember with quantum memory~\cite{huang2021information, aharonov2021quantum, chen2021exponential, huang2021quantum, chen2024tight, oh2024entanglement, allen2025quantum, liu2025quantum}. Such quantum learning agents signify a paradigm shift beyond traditional quantum sensing approaches by enabling coherent processing of quantum information across multiple experimental runs.

\vspace{0.5em}\paragraph*{Quantum-computing-enhanced sensing.} When applied to sensing classical oscillating fields such as electromagnetic signals or gravitational waves, quantum agents could achieve beyond-quadratic polynomial speedups. While quantum algorithms like Grover search~\cite{grover1996fast} were not traditionally considered relevant to sensing applications, recent work~\cite{allen2025quantum} demonstrates how quantum computers can be integrated with quantum sensors in a framework termed quantum-computing-enhanced sensing to learn properties of classical fields substantially faster than quantum sensors controlled by classical computers alone.

\vspace{0.5em}\paragraph*{Exponential learning advantages.} While the aforementioned advantages are polynomial, one may wonder whether exponential quantum advantages exist in the learning and sensing realm. The answer is proven to be affirmative when quantum learning agents are applied to study many-body physical systems and dynamics~\cite{huang2021information, aharonov2021quantum, chen2021exponential, huang2021quantum, chen2024tight, oh2024entanglement, liu2025quantum}. These quantum agents can learn from exponentially fewer experiments than classical approaches for tasks including predicting physical properties and learning approximate models of quantum dynamics. The exponential improvement emerges from the ability to preserve quantum correlations across multiple experiments that classical approaches must necessarily destroy through intermediate measurements, enabling quantum learning agents to extract exponentially more information per experiment than their classical counterparts.

\vspace{1em}
Rather than relying on unproven complexity assumptions, quantum advantages in learning/sensing realm often emerge from the quantum nature of physical reality itself. Looking ahead, the most transformative opportunities may emerge from quantum learning agents that seamlessly integrate quantum sensors and quantum computers. Some of these advantages have been proven to persist under realistic noise~\cite{huang2021quantum, huang2022foundations, chen2023complexity, liu2025quantum}, but the general theory of fault-tolerant quantum agents remains to be developed.

\subsection{Cryptographic/communication/strategic game advantages}

In contrast to both computational and learning/sensing advantages, the cryptographic realm encompasses quantum advantages that emerge from quantum mechanics' fundamental properties, such as entanglement, no-cloning, and measurement disturbance, applied to security, communication, and coordination tasks. Unlike computational advantages that depend on unproven complexity assumptions, these advantages derive their power directly from the mathematical structure of quantum mechanics itself, providing unconditional security guarantees that remain inviolable as long as quantum mechanics remains an accurate description of our universe.

\vspace{0.5em}\paragraph*{Physics-based security.} The defining characteristic of this realm is that quantum mechanics itself provides security guarantees impossible to achieve classically. The no-cloning theorem ensures that unknown quantum states cannot be accurately copied without detection. The measurement postulate guarantees that extracting information from a quantum system necessarily disturbs it, leaving unmistakable evidence if someone has been eavesdropping. Quantum entanglement enables correlations that fundamentally surpass classical limitations, as demonstrated by Bell's inequality. Together, these principles create a foundation for cryptographic protocols whose security is guaranteed by the laws of physics.

\vspace{0.5em}\paragraph*{Quantum encryption and certified randomness.} Quantum mechanics enables parties to establish secure communication channels and verify information authenticity with guarantees that transcend computational security. Quantum key distribution exemplifies this revolutionary capability: the BB84 protocol~\cite{bennett1984quantum} enables Alice and Bob to establish a shared cryptographic key via quantum communication with security guaranteed by the laws of physics. Any eavesdropper necessarily disturbs the quantum transmission, revealing their presence through detectable changes in the quantum states. This represents a paradigmatic shift from classical cryptographic protocols, where security relies on the assumption that computations that could break the protocol in principle are too difficult to execute in practice. Such computational assumptions could potentially be undermined by future algorithmic breakthroughs or advances in computational power. Certified randomness generation~\cite{colbeck2012quantum, liu2025certified} exploits these principles to create truly random bits that have not previously been generated prior to the protocol, whose unpredictability is guaranteed by quantum mechanics. Unlike classical pseudorandom number generators that are fundamentally deterministic, quantum randomness provides bits that cannot be spoofed by a secret seed or other backdoor. Furthermore, the randomness is guaranteed even if the user does not trust the inner workings of the quantum device used in the protocol. Such device-independent randomness certification has profound implications for cryptographic applications requiring secure random bits.

\vspace{0.5em}\paragraph*{Unforgeability protocols.} The no-cloning theorem enables the creation of objects with inherent copy-protection, opening new classes of applications impossible in classical physics. Quantum money schemes~\cite{wiesner1983conjugate} aim to create unforgeable currency by using quantum states to encode serial numbers that can be efficiently verified but cannot be perfectly duplicated. Certified deletion protocols~\cite{broadbent2020uncloneable} enable Alice to send quantum-encrypted data to Bob, who can later provide cryptographic proof that he has permanently and verifiably destroyed the information. This capability is impossible with classical physics, where information can always be copied before deletion. Copy-protected software~\cite{aaronson2009quantum} and one-time programs~\cite{goldwasser2008one} extend these principles to software distribution, enabling programs that cannot be duplicated while preserving their computational functionality. These protocols may revolutionize intellectual property protection, data privacy, regulatory compliance, and rights management by making such violations physically impossible rather than merely legally prohibited.

\vspace{0.5em}\paragraph*{Quantum entanglement and strategic games.} Quantum entanglement transcends classical communication limitations through nonlocal correlations that exhibit behaviors impossible in any local classical theory. As explored in Puzzle~\ref{puzzle:spooky}, entangled quantum systems demonstrate correlations that cannot be reproduced by any classical mechanism. In strategic scenarios requiring spatially separated parties to act simultaneously in a coordinated manner, quantum protocols could provide fundamental advantages over classical protocols. For example,~\cite{ding2024coordinating} has illustrated how entangled quantum systems shared between spatially separated stock exchanges can provide coordination advantages in some high-frequency trading applications, where small time improvements can translate into significant competitive value.

\vspace{1em}
The cryptographic, communication, and strategic game realm uses quantum mechanics to provide security guarantees that are categorically impossible in the classical world. Unlike computational advantages that may be undermined by unexpected algorithmic discoveries, these advantages are guaranteed by the mathematical structure of quantum mechanics. As quantum technologies mature, this realm could revolutionize how we approach security, privacy, and trust in our interconnected world.

\subsection{Space advantages}

The realm of space advantages addresses a fundamental challenge of our data-driven era: processing data streams that vastly exceed available memory. While many classical algorithms require memory proportional to input size, quantum algorithms can sometimes achieve exponential space savings by encoding classical information in the exponentially large Hilbert space of a quantum device.

\vspace{0.5em}\paragraph*{Quantum compression for data streams.} As demonstrated in Puzzle~\ref{puzzle:Holevo-vs-Raz}, quantum states can serve as exponentially compressed representations of classical data for specific computational tasks. While Holevo's bound~\cite{holevo1973bounds} forbids the extraction of an exponential number of classical bits from a small quantum memory, certain computational tasks can be performed on the compressed quantum representation without full data reconstruction. This quantum compression becomes particularly powerful in streaming scenarios where massive data streams must be processed with limited memory. Classical algorithms must make irreversible decisions about which information to retain as each data element arrives. In contrast, quantum algorithms can maintain superpositions over many possible configurations simultaneously, effectively postponing these irreversible decisions until more information becomes available.

\vspace{0.5em}\paragraph*{Applications and limitations.} Recent work on directed maximum cut problems~\cite{kallaugher2024exponential} demonstrates how quantum algorithms can solve natural graph problems using exponentially less memory than classical approaches when the input size vastly exceeds memory capacity. The vector-in-subspace problem studied by Raz~\cite{raz1999exponential} provides another example: determining which of two subspaces contains a given vector requires exponential classical communication but only linear quantum communication. However, space advantages face important limitations. They typically apply only when input size significantly exceeds memory capacity; furthermore, inputting a large amount of data into a quantum system can be costly, and the computational task performed must be compatible with Holevo's bound. Despite these constraints, space advantages represent a promising direction as data generation accelerates.

\vspace{1em}
This realm highlights how quantum superposition can address fundamental resource limitations in data processing and storage, offering more efficient approaches to existing computational tasks. While this realm remains the least explored, it holds significant potential for making previously infeasible large-scale computations tractable through exponential memory reductions.

\section{The future of quantum advantage}

\noindent The past decades have witnessed remarkable progress in predicting quantum advantages through theoretical approaches. Using mathematical proofs and complexity-theoretic arguments, researchers have established quantum advantages for various computational, sensing, learning, and communication tasks. Yet this theoretical framework, while powerful, captures only part of the story. The history of classical computing offers a cautionary tale: algorithms with the strongest provable advantages rarely become the most useful in practice, and the most transformative applications of modern computers were often unpredicted or achieved without rigorous performance guarantees.  This observation suggests we need a broader framework for identifying and understanding genuine quantum advantages.

Envision a future, perhaps 100 years from now, where quantum technology has become as ubiquitous as classical computers are today. A complete quantum technological infrastructure would enable us to probe the world with quantum sensors, preserve quantum information in reliable quantum memories, process quantum data with fault-tolerant quantum computers, and transmit quantum states through quantum networks. In this future, the competition between quantum and classical approaches would shift from theoretical guarantees to head-to-head empirical competition, much as classical algorithms compete today based on practical metrics rather than asymptotic bounds. How should we expand our framework for understanding such future quantum advantages?

In this section, we explore three perspectives that we expect to become increasingly more relevant in the future. First, we examine how empirical evidence, rather than theoretical guarantees, might drive the discovery and validation of quantum advantages. Second, we consider how quantum approaches could offer conceptual advantages, new ways of thinking about problems that turn out to be more natural or insightful than classical approaches. Finally, and most profoundly, we show that some quantum advantages are fundamentally unpredictable using classical resources alone, suggesting that the landscape of quantum advantages may be far richer than we can currently envision.

\subsection{Empirical quantum advantages}

\noindent
First we emphasize that empirical performance, rather than theoretical guarantees, will likely drive the adoption of quantum methods in practice. Key metrics will include empirical runtime, cost of implementation, and total time-to-solution. This shift mirrors the evolution of classical computing, where widely used algorithms often differ from those with the strongest asymptotic guarantees. When focusing on runtime, empirical advantages can manifest in multiple ways. The most straightforward are constant-factor improvements on practical problem instances; these may not grab the attention of theorists but may be very important to users. It is also notable that the classical algorithms achieving the best asymptotic scaling might not be the most practically useful due to prohibitively large constant factors. Multiplication of $n\times n$ matrices provides a classic example. Although theoretically superior algorithms exist with better asymptotic complexity, the simple algorithm with $\mathcal{O}(n^3)$ steps remains standard in practice due to its smaller constants and cache efficiency.

Quantum advantage may also emerge from faster total time-to-solution starting from the problem or concept going all the way to the solution, which we term ``concept-to-solution''. Concept-to-solution includes not only the algorithm's runtime, but the entire process from problem formulation to working solution. This encompasses the human time spent designing experimental protocols, implementing algorithms, and iteratively refining them. A practical quantum advantage might occur if quantum technology enables faster overall solution time, even if it turns out that classical algorithms with similar runtime exist but are not discovered until much later.

\begin{example}[Quantum advantage in concept-to-solution]
Consider a scenario where designing a quantum algorithm for a new computational problem takes an hour and the quantum computation itself solves the problem in about one second. Even if a talented computer scientist could eventually discover a classical algorithm with a one-second classical computational time, if that discovery requires a year of dedicated research, the quantum approach still offers a significant advantage in concept-to-solution: an hour versus a year. This form of advantage, where quantum methods provide a more direct path to efficient solutions, may prove particularly valuable as quantum technology matures.
\end{example}

Large-scale quantum computing might also reveal new empirical benefits that transcend theoretical guarantees. Modern machine learning offers a preview of such possibilities. Large language models have demonstrated capabilities far beyond theoretical predictions. Similarly, as quantum systems scale up in size and architectural complexity, they may exhibit unexpected emergent behaviors that, while difficult to characterize theoretically, provide transformative practical advantages.

Such empirical quantum advantages represent a significant departure from how we typically evaluate quantum technology today. Rather than focusing solely on asymptotic complexity and provable speedups, we envision a future where quantum methods might be preferred for their practical benefits: faster solution development, simpler implementation, or emergent capabilities that arise from scale. While such advantages may be harder to prove rigorously, they could ultimately drive the widespread adoption of quantum technology.

\subsection{Conceptual quantum advantages}

\noindent Secondly, we highlight the value of looking beyond physical resources like time, space, cost, and energy to consider broader forms of advantage. Modern philosophy of science recognizes that empirically equivalent theories can differ dramatically in their conceptual value. The heliocentric model of the solar system, while making essentially identical predictions to certain geocentric models, provided a simpler framework that transformed our understanding of the cosmos and catalyzed new scientific directions.

We propose that quantum technology might offer similar conceptual advantages. Consider two solutions to a problem, one classical, one quantum, with similar resource requirements. Even if they offer comparable performance, the quantum solution might be conceptually simpler, providing a more natural framework for understanding the problem. This conceptual clarity could manifest in multiple ways: easier reasoning about correctness, more straightforward integration with other components, or more natural extensions to related problems.

The value of such conceptual advantages extends beyond mere aesthetic preference. Simpler conceptual frameworks often reflect deeper underlying structure, potentially leading researchers to new theoretical insights and technological directions. Just as quantum mechanics itself provided new ways of understanding physical phenomena, quantum approaches to practical problems might offer new perspectives that catalyze further advances. While such advantages might seem secondary in today's resource-constrained quantum era, they could become crucial in a future with ubiquitous quantum technology, forming cornerstones for new theoretical and technological developments.

\subsection{Unpredictable quantum advantages}

\noindent Our third and most profound perspective challenges a core assumption of our current framework: that we can predict quantum advantages using classical resources and mathematical reasoning. We prove that this assumption fails even for well-defined computational problems, revealing fundamental limits to our ability to identify quantum advantages using classical approaches alone. This result illustrates concretely the more general principle that some emergent quantum advantages may be quite unexpected when they are discovered.

To establish this limitation rigorously, we consider a concrete computational problem regarding whether a specified computational task admits a quantum advantage or not.  This problem is inspired by the field of meta-complexity, the study of the complexity of computational problems which address questions about complexity theory itself \cite{Ren2023relativization, Simons2023meta}. Earlier, we discussed how Pauli propagation provides an efficient classical method for simulating most quantum circuits, though large families of circuits remain classically intractable. This naturally leads us to ask: given a specific quantum circuit, can we determine whether executing it on a quantum computer would outperform this classical simulation method?

\begin{theorem}[Classical hardness in predicting quantum advantages; informal]
Consider the decision problem of predicting whether executing a particular quantum circuit on a quantum computer has a computational advantage over the Pauli propagation classical simulation method on that same circuit. Assuming $\mathsf{BPP} \neq \mathsf{BQP}$, this decision problem exhibits a quantum advantage: while a quantum computer can efficiently solve this problem, no classical algorithm can solve it efficiently.
\end{theorem}

This theorem reveals a profound circularity in our quest for quantum advantages: the very act of predicting quantum advantages is itself a problem requiring quantum computation. Since Pauli propagation represents only one classical simulation method, predicting quantum advantage against any classical method, including those not yet conceived, is even more challenging. The proof, provided in Appendix~\ref{app:hardness-predict-qadv}, shows that any classical algorithm that can determine efficiently whether or not running a given quantum circuit on a quantum computer has a quantum advantage over a classical heuristic can be adapted to create a classical algorithm that is able to efficiently solve any problem in $\mathsf{BQP}$. Thus, detecting quantum advantage is classically hard, assuming that $\mathsf{BPP} \neq \mathsf{BQP}$. Furthermore, determining whether executing a quantum circuit exhibits quantum advantage is quantumly easy; it suffices to run both the quantum circuit and the classical heuristic on a polynomial number of randomly chosen instances and then check whether or not the results typically agree. 

This establishes a fundamental barrier: classical technology cannot reliably distinguish genuine quantum advantages from pseudo-advantages, confirming that quantum advantages exist beyond the predictive power of classical computers. This leads us to a profound implication:
\begin{center}
\textit{There are plenty of problems with genuine quantum advantages such that\\ we cannot predict the advantage using only classical technology.}
\end{center}
Paradoxically, to fully map out the landscape of quantum advantages, we must use the very quantum technologies whose power we are trying to characterize.

This unpredictability extends beyond computational problems to the frontiers of scientific discovery. Quantum technologies are already kindling new insights into the behavior of highly entangled quantum many-body systems; recent examples include quantum many-body scars and deep thermalization \cite{bernien2017probing, turner2018weak, ippoliti2022solvable, cotler2023emergent}. What marvelous phenomena are yet to be found? We cannot possibly know if our classical technologies are intrinsically incapable of predicting the phenomena. Thus a rich panoply of surprises may be unveiled as we probe the physical universe with ever more advanced quantum devices.

\subsection{The path forward}

We look forward to a future where a broad spectrum of quantum advantages come to light: some empirically discovered, others conceptually transformative, and many fundamentally unpredictable using current technology. Just as the most impactful applications of classical computers were unimaginable to their early pioneers, the true potential of quantum technology might only reveal itself through a journey of development and discovery.

We are not in that future yet. Today's quantum technology is limited, and mathematical analysis provides our best guide for identifying genuine quantum advantages. While mathematical rigor will always be essential, we have argued that theoretical analysis alone cannot capture the full landscape of potential quantum advantages awaiting discovery. For now, we face the challenge of developing new ideas and methods to navigate the vast world of quantum advantage, guided by, but not confined to, what is mathematically provable.

\vspace{0.5em}
\subsection*{Acknowledgments:}
\vspace{-0.5em}
The authors thank Sergio Boixo, Ryan Babbush, Edward Farhi, Sam Gutmann, and Stephen Jordan for a careful reading and helpful comments. We would also like to thank Robbie King and Mikhail Lukin for illuminating discussions, and Chi-Yun Claudia Cheng for discussions and assistance regarding the main figure. The authors acknowledge Mauro Schiulaz and other editors from Physical Review X for insightful discussions that led to the inception of this perspective article.
H.H. acknowledges support from the U.S. Department of Energy, Office of Science, National Quantum Information Science Research Centers, Quantum Systems Accelerator and the Broadcom Innovation Fund.
S.C. acknowledges the support from the Center for Ultracold Atoms, an NSF Physics Frontiers Center (Grant No. 2317134).
J.P. acknowledges support from the U.S. Department of Energy, Office of Science, Accelerated Research in Quantum Computing, Quantum Utility through Advanced Computational Quantum Algorithms (QUACQ) and Fundamental Algorithmic Research toward Quantum Utility (FAR-Qu), from the U.S. Department of Energy, Office of Science, National Quantum Information Science Research Centers, Quantum Systems Accelerator, and from the Institute for Quantum Information and Matter, an NSF Physics Frontiers Center (PHY-2317110).

\cleardoublepage

\begingroup
\setlength{\spaceskip}{0.5em}

\begin{center}
    {\fontsize{24}{30}\selectfont \textbf{Appendices}}
\end{center}

\vspace{0.5cm}
\hrule height 0.5pt
\vspace{0.25cm}
\endgroup

\appendix

\section{Solutions to Puzzles}

\subsection{Puzzle~\ref{puzzle:spooky}: Is this spooky?}
\label{app:spooky-puzzle}

Let us analyze this puzzle in two steps: first considering measurements in a single basis, then examining what changes when we allow measurements in multiple bases.

\vspace{0.5em}\paragraph*{First scenario: Single-basis measurements.} When Alice and Bob only measure in the standard basis $\{\ket{\uparrow}, \ket{\downarrow}\}$, their quantum state
\begin{equation}
    \frac{1}{\sqrt{2}}(\ket{\uparrow} \otimes \ket{\uparrow} - \ket{\downarrow} \otimes \ket{\downarrow})
\end{equation}
produces perfectly correlated outcomes. This correlation appears mysterious. When Alice measures her spin and obtains $\ket{\uparrow}$, Bob's spin instantly collapses to $\ket{\uparrow}$, and similarly for $\ket{\downarrow}$. However, this behavior can be perfectly simulated using classical objects. Consider this classical protocol:
\begin{enumerate}
    \item Take a pair of socks that are either both black or both white with equal probability.
    \item Place each sock in an opaque gift box.
    \item Give one box to Alice and one to Bob.
    \item Let them separate to arbitrary distances.
\end{enumerate}
When Alice and Bob open their boxes, they see perfectly correlated colors instantaneously. The classical protocol exactly reproduces the quantum statistics: each party sees a random outcome (up/down or black/white) with 50-50 probability, and the outcomes are always perfectly correlated.

\vspace{0.5em}\paragraph*{Second scenario: Multiple-basis measurements.} The situation changes dramatically when Alice and Bob can each randomly choose between measuring in the standard basis ${\ket{\uparrow}, \ket{\downarrow}}$ or the rotated basis ${\ket{\nearrow}, \ket{\swarrow}}$. Now they observe correlations that cannot be reproduced by any classical system. While the measurement procedures themselves are quantum, the remarkable feature is that the resulting statistical patterns in their classical outcomes, i.e., the sequences of bits they record, exhibit correlations that no classical mechanism can generate. The classical impossibility was first proven by Bell~\cite{bell1964einstein} and has been experimentally verified multiple times~\cite{aspect1982experimental,freedman1972experimental}. The correlations exhibited by quantum entanglement in multiple measurement bases violate Bell's inequality, providing rigorous proof that quantum systems can achieve statistical correlations among classical outcomes that are forbidden to any classical theory. This elucidates a genuine quantum advantage. We refer readers interested in a detailed pedagogical treatment of Bell's inequality and its implications to Chapter 4 of \cite{preskill1998lecture}.

\subsection{Puzzle~\ref{puzzle:beat-qram}: Exponential quantum advantage in analyzing big data?}\label{app:beat-qram}

The solution to Puzzle~\ref{puzzle:beat-qram} reveals how classical techniques can match quantum computation for this specific task. The key insight is that Bob's classical random access memory (RAM) can be enhanced to provide two crucial capabilities, each taking only $\mathrm{poly}(n)$ time:
\begin{enumerate}
    \item \emph{Query:} Given any indices $k$ and $i$, output the value $x_{k,i}$.
    \item \emph{Sample:} Given index $k$, sample $i$ with probability $p_k(i) = |x_{k,i}|^2$.
\end{enumerate}
The construction of efficient sampling access from standard RAM is a well-known technique in randomized algorithms~\cite{devroye1986non, vose1991linear}. The goal is to sample from a discrete probability distribution where outcome $i$ has probability $p_k(i) = |x_{k,i}|^2$. The naive approach would require $\mathcal{O}(2^n)$ time per sample by iterating through all entries, but preprocessing enables much faster sampling.

For each vector $\vec{x}_k$, we first compute the cumulative distribution function:
\begin{equation}
F_k(j) = \sum_{i=1}^j |x_{k,i}|^2.
\end{equation}
Since $\vec{x}_k$ is normalized, we have $F_k(2^n) = 1$. We then store these cumulative probabilities in a binary tree data structure where each leaf corresponds to an index $i \in \{1, 2, \ldots, 2^n\}$, and each internal node stores the sum of probabilities in its subtree. To sample from this distribution, we generate a uniform random number $r \in [0,1]$ and use binary search to find the unique index $i$ such that $F_k(i-1) < r \leq F_k(i)$. The binary tree structure allows us to traverse from the root to the appropriate leaf in $\mathcal{O}(\log 2^n) = \mathcal{O}(n)$ time: at each internal node, we compare $r$ with the left subtree's total probability and branch left or right accordingly. This preprocessing requires $\mathcal{O}(2^n)$ time per vector to compute all cumulative probabilities and construct the tree, similar to the overhead required by quantum random access memory (QRAM). Once completed, this approach enables $\mathcal{O}(n)$ time sampling, making the sampling step efficient.

Using these capabilities, Bob can employ importance sampling to estimate the inner product efficiently. The key idea is to define a random variable $X_{k,\ell}$ as follows: first sample $i$ according to probability $p_k(i)$, then let $X_{k,\ell} = \frac{x_{\ell,i}}{x_{k,i}}$. A straightforward calculation shows that:
\begin{align}
    \E[X_{k,\ell}] &= \sum_{i=1}^{2^n} p_k(i) \cdot \frac{x_{\ell,i}}{x_{k,i}} = \sum_{i=1}^{2^n} |x_{k,i}|^2 \cdot \frac{x_{\ell,i}}{x_{k,i}} = \vec{x_k} \cdot \vec{x_{\ell}}\\
    \mathrm{Var}[X_{k,\ell}] &= \sum_{i=1}^{2^n} p_k(i) \cdot \left(\frac{x_{\ell,i}}{x_{k,i}}\right)^2 - (\E[X_{k,\ell}])^2 \leq 1
\end{align}
The bounded variance is crucial: it means Bob can estimate the inner product to precision $1/\mathrm{poly}(n)$ by averaging $\mathrm{poly}(n)$ independent samples of $X_{k,\ell}$. This classical randomized algorithm matches the performance of the quantum algorithm exactly.

This solution exemplifies a broader pattern in quantum algorithm analysis. Similar classical techniques have appeared across different fields, including neural quantum states~\cite{carleo2017solving}, distributed computing for inner product estimation~\cite{buhrman2010nonlocality}, randomized linear algebra~\cite{mahoney2011randomized}, and the dequantization of quantum algorithms~\cite{tang2019quantum}. These examples demonstrate how classical probabilistic techniques can sometimes achieve the same performance as quantum algorithms when operating under comparable data-access assumptions.

\subsection{Puzzle \ref{puzzle:Holevo-vs-Raz}: Can quantum systems encode exponential classical information?}
\label{app:Holevo-vs-Raz}

The answer to this puzzle is a nuanced ``yes,'' that reveals a subtle but powerful form of quantum advantage. The apparent conflict between exponential encoding and Holevo's bound is resolved by understanding that the usefulness of an encoding depends entirely on the task for which the information is used. Holevo's bound \cite{holevo1973bounds} states that it is impossible for Bob to learn all $2^n$ classical values $x_i$ from the following $n$-qubit state:
\begin{equation}
    \ket{x} = \sum_{i=1}^{2^n} x_i \ket{i}.
\end{equation}
Any quantum measurement strategy he employs can reveal at most $\mathcal{O}(n)$ bits of classical information about the $2^n$-dimensional vector $\vec{x}$. If Alice's goal were to transmit the entire classical vector to Bob, then quantum state encoding offers no asymptotic advantage over sending $n$ classical bits.

However, the story changes if Bob's goal is not to reconstruct the entire vector, but to perform a specific computation on it. For some computational tasks, the quantum state encoding provides a genuine, exponential advantage in communication. A landmark result by Raz \cite{raz1999exponential} established this rigorously. He defined a relational problem where Alice is given a unit vector $v \in \mathbb{R}^{2^n}$ and Bob is given one of two possible subspaces, $P_0$ or $P_1$, each of dimension $2^{n}/2$. The vector $v$ is promised to lie in one of these two subspaces, and their task is to determine which one. Raz proved that any classical communication protocol for this problem requires exchanging at least $\Omega(2^{n/2})$ bits. In contrast, a quantum protocol can solve the problem with exponentially less communication: Alice simply encodes her vector into an $n$-qubit state $\ket{v}$ and sends it to Bob, who can then perform a measurement to identify the correct subspace. This establishes an exponential separation between the $\mathcal{O}(n)$ quantum communication complexity and the $\Omega(2^{n/2})$ classical communication complexity for this particular task.

This example demonstrates that for certain problems, the quantum state $\ket{x}$ acts as an exponentially compressed ``quantum fingerprint'' of the classical data. While this fingerprint doesn't allow full reconstruction of the original data, it contains enough information to answer specific questions about it exponentially more efficiently than any classical representation of similar size.
Therefore, Alice's quantum compression scheme is not an illusion. It provides a genuine quantum advantage, but this advantage is in \emph{communication complexity} for specific relational problems, not in general-purpose data transmission. This principle has found applications in quantum machine learning, for example, in protocols for learning quadratic functions of large datasets, where quantum communication can be exponentially more efficient than classical communication \cite{gilboa2024exponential}. Alice and Bob did indeed gain an advantage, provided their future data processing tasks are of the right kind.

\section{Fragility of quantum advantage in noisy sensing} \label{app:no_heisenberg}

Here we illustrate the claims in Example \ref{ex:sensing} in a simple model of quantum sensing. We will explain how entanglement can be exploited to improve sensitivity in the noiseless case, and prove that this sensitivity enhancement does not survive in the presence of noise.

\subsection{Entanglement-enhanced sensitivity}

In the noiseless version of this model, our task is to distinguish between two single-qubit unitary channels.  $\mathcal{C}_\theta$ rotates the qubit about the $Z$-axis by the positive angle $\theta$, and $\mathcal{C}_0$ acts trivially:
\begin{align}
    \mathcal{C}_\theta [\rho] = e^{-i \frac{\theta}{2} Z} \rho e^{i \frac{\theta}{2} Z}\;\;\; \textrm{and} \;\;\;
    \mathcal{C}_0 [\rho] = \rho.
\end{align}
The rotation angle $\theta$ is the strength of the signal we wish to detect. Our goal is to detect the signal as quickly as possible; that is, to distinguish between the two cases using as few resources as possible. We will measure sensing efficiency in terms of the \emph{sensing time} $T$, defined as the wall-clock time required to complete the protocol assuming one channel use takes one unit of time. Alternatively, we may fix the sensing time $T$ and seek the smallest signal strength~$\theta$ that can be distinguished from zero. This minimal detectable signal strength is called the \emph{sensitivity}. We are most interested in the regime $0 < \theta \ll 1$, which will be assumed below.

When only a single sensing qubit is available, the optimal sensing protocol is well understood: prepare the qubit in a superposition state $\ket{+} = \frac{1}{\sqrt{2}} \left( \ket{0} + \ket{1} \right)$, apply the channel $T \approx \pi/\theta $ times sequentially, and measure the qubit in the $X$-basis.  Under $\mathcal{C}_0$, the qubit remains in the initial state, and the measurement yields $\ket{+}$ deterministically.
Under $\mathcal{C}_\theta$, the qubit rotates by the total angle $T \theta \approx \pi$, and the measurement produces $\ket{-}$ with high probability, enabling us to distinguish the two cases.
This procedure can be repeated a few times to amplify our confidence.
Thus, one can resolve a small angle $\theta$ within a sensing time of $T = \mathcal{O}(1/\theta)$.

This analysis directly generalizes to $N$ sensing qubits working in parallel. We now consider the channel $\mathcal{C}_x^{\otimes N}$ with $x=0$ or $\theta$, which applies the single-qubit channel $\mathcal{C}_x$ to all $N$ qubits simultaneously.
In this case, one can readily shorten the sensing time to $T=\mathcal{O}\left( \frac{1}{\sqrt{N} \theta }\right)$ by performing $N$ copies of the single-particle protocol in parallel.
Here, $T$ represents the number of times the parallel channel $\mathcal{C}_x^{\otimes N}$ is applied, which equals the wall-clock time since all $N$ qubits are processed simultaneously in each time step.
More explicitly, each particle is prepared in the $\ket{+}$ state, the channel $\mathcal{C}_x^{\otimes N}$ is applied $T =  \lceil \frac{\pi }{\sqrt{N} \theta}\rceil$ times, and then all qubits are measured.
Recall that $T$, the number of (parallel) channel uses, is an integer; we are implicitly considering the case where $\sqrt{N}\theta\ll 1$.
The measurement of each particle yields $\ket{-}$ with probability $\Omega\left( 1/N\right)$ when $x=\theta$. Since a single $\ket{-}$ outcome heralds the presence of the signal, the two channels can be distinguished with high confidence after a constant number of repetitions.
Thus, this protocol achieves a sensitivity scaling as
\begin{align}
    \theta^{-1} = \mathcal{O} \left( \sqrt{N} T\right),
\end{align}
which one can show is optimal for any protocol that does not use entangled states.

Using entangled quantum states, the sensing time can be further improved to $T=\mathcal{O}\left( \frac{1}{N\theta}  \right)$, where we now assume $N\theta \ll 1$.
Consider the GHZ state $\ket{\textrm{GHZ};+}= \left( \ket{0}^{\otimes N} + \ket{1}^{\otimes N}\right)/\sqrt{2}$ evolved under the parallel channel $\mathcal{C}_\theta^{\otimes N}$ repeated $T=\lceil \frac{\pi}{N\theta} \rceil$ times.
The resulting state has a large overlap with $\ket{\textrm{GHZ};-}$, which is orthogonal to the initial state.
Thus, by performing measurements in the $\{\ket{\textrm{GHZ};\pm}\}$ basis, one can distinguish the two channels within time $T=\mathcal{O}\left( \frac{1}{N\theta} \right)$.
The reduction in sensing time by the factor $\sqrt{N}$ translates to an improved sensitivity arising from the multipartite entanglement in the GHZ state:
\begin{align}
    \theta^{-1} = \mathcal{O}\left( N T \right).
\end{align}
One can show that this is the optimal sensitivity scaling using $N$ sensing qubits. This entanglement-enhanced sensitivity is an example of quantum sensing advantage.

\subsection{Fragility of entanglement-enhanced sensitivity in noisy sensors}

We now consider a simple model of noisy quantum sensing. In realistic settings, noise arises from a variety of physical sources such as thermal fluctuations in the environment, imperfections in laser beams or microwave control pulses, and intrinsic physical properties of atoms or molecules. In our simplified model, sensing qubits are subject to rotation by an angle $\theta + \theta_n$, where $\theta$ is the signal to be detected and $\theta_n$ represents noise. The signal strength $\theta$ is assumed to be small and has the same value every time $\mathcal{C}_\theta$ is applied. The noise $\theta_n$ afflicts both $\mathcal{C}_\theta$ and $\mathcal{C}_0$ and is randomly fluctuating. We assume that $\theta_n$ is drawn independently from a Gaussian distribution with mean zero and variance $\gamma$ each time either channel is applied. Our task is to determine whether $\theta = 0$ or not.

More explicitly, we need to distinguish between two channels: 
\begin{align}
     \tilde{\mathcal{C}}_{\theta} &\equiv \mathbb{E}_{\theta_n} \left[ \mathcal{C}_{\theta+\theta_n}\right] = ((1-p)\mathcal{I} + p \mathcal{D}) \circ \mathcal{C}_\theta,\\
    \tilde{\mathcal{C}}_0 &\equiv \mathbb{E}_{\theta_n} \left[ \mathcal{C}_{\theta_n}\right] = ((1-p)\mathcal{I} + p \mathcal{D})  \circ \mathcal{C}_0,
\end{align}
where $p = 1-e^{-\gamma/2}$ is the probability that the qubit dephases, $\mathcal{I}$ is the identity channel $\mathcal{I}[\rho] = \rho$, and $\mathcal{D} [\rho] = \ketbra{0}{0}\rho \ketbra{0}{0} + \ketbra{1}{1} \rho \ketbra{1}{1}$ is the completely dephasing channel in the $Z$-basis. In other words, the density operator's diagonal elements in the $Z$-basis are unaffected by the Gaussian dephasing noise, while off-diagonal elements are suppressed by the factor $e^{-\gamma/2}$. We are especially interested in the high-sensitivity regime where $\theta \ll \gamma$.

For this noisy quantum sensing model, a lower bound on the sensing time is captured by the following theorem.

\begin{theorem}[Noisy small-angle detection]\label{thm:noisy_small_angle_detection}
Any quantum algorithm that correctly distinguishes two quantum channels $\tilde{\mathcal{C}}_0^{\otimes N}$ and $\tilde{\mathcal{C}}_\theta^{\otimes N}$ with probability $> 2/3$ must query the channel at least $T$ times, where  
\begin{align}
    NT = \Omega\left(\max\left( \frac{\gamma}{\theta^2}, \frac{1}{\theta} \right)\right).
\end{align}
We note that quantum algorithms may involve arbitrarily many and arbitrarily fast gate operations as well as an unlimited number of ancilla qubits and unbounded memory coherence times.
\end{theorem} 

\noindent 
Theorem \ref{thm:noisy_small_angle_detection} implies that for any small constant noise strength $\gamma$, the sensing time $T$ scales inverse quadratically in the signal strength $\theta$. Thus, for any constant $\gamma$, the asymptotic sensitivity scales as 
\begin{align}
    \theta^{-1} = \mathcal{O} \left( \sqrt{NT}\right).
\end{align}
Importantly, this optimal sensitivity scaling can be attained by a simple, fully separable protocol that requires no entanglement, demonstrating that entanglement does not offer an asymptotic sensitivity advantage in this noisy quantum sensing model. Intuitively, dephasing noise destroys the phase coherences that enable entangled states like GHZ to accumulate phase collectively, eliminating the quantum advantage observed in the noiseless case.

The simple, fully separable protocol works as follows. We focus on the regime $\theta, \gamma \ll 1$ and $\theta \ll \gamma$ to simplify our analysis. Each sensing qubit is initialized in the $\ket{+}$ state, undergoes evolution under the channel $\lceil 1/\gamma \rceil$ times, and is then measured in the $Y$-basis, yielding one of two outcomes:
$\ket{\pm i} = \frac{1}{\sqrt{2}} \left( \ket{0} \pm i\ket{1} \right)$.
This procedure is repeated $K$ times for $N$ qubits in parallel, producing a total of $KN$ measurement outcomes within sensing time $T = K \lceil 1/\gamma \rceil$. Finally, one counts the fraction of $\ket{+i}$ outcomes, $\hat{f}$, among all measurements.

In the absence of the signal, both outcomes $\ket{\pm i}$ occur with exactly the same probabilities, hence the probability distribution of $\hat{f}$ is centered around $\frac{1}{2}$ with variance $\frac{1}{4KN}$. In the presence of the signal, the measurement yields $\ket{\pm i}$ with probabilities $\frac{1}{2} \pm \epsilon$ with positive bias $\epsilon = \sin (\theta \lceil 1/\gamma \rceil) e^{-\gamma \lceil 1/\gamma \rceil/2}/2$, and the distribution of $\hat{f}$ is centered around $\frac{1}{2} + \epsilon$ with variance $\frac{1-4\epsilon^2}{4KN}$. Therefore, by checking whether $\hat{f}$ is larger than a threshold $\frac{1}{2} +\frac{\epsilon}{2}$, one can distinguish the two cases. The failure probability can be bounded using Chebyshev's inequality and becomes small whenever the standard deviation is sufficiently small compared to the bias: $\epsilon \gg \frac{1}{\sqrt{4KN}}$, which implies $NT = \mathcal{O} \left( \frac{e^{\gamma \lceil 1/\gamma \rceil} / \lceil 1/\gamma \rceil}{\theta^2}\right)$ is sufficient. This matches the optimal sensitivity scaling for any fixed constant $\gamma$.

We now proceed to the proof of Theorem \ref{thm:noisy_small_angle_detection}.

\begin{proof} 
We prove two lower bounds separately: (i) $NT = \Omega (\gamma/\theta^2)$ and (ii) $NT = \Omega(1/\theta)$.
The second bound follows directly from the noiseless case. This can be shown by noting that access to noiseless channels is at least as powerful as access to noisy channels because the latter can be simulated from the former by artificially adding dephasing channel after the noiseless channel. As a result, the sensing time for the noisy case cannot be shorter than the noiseless case, which requires $NT = \Omega(1/\theta)$.

To prove the first bound, we employ a reduction to classical hypothesis testing. We introduce a hypothetical agent called \emph{Demigod}, who is at least as powerful as any quantum agent capable of running arbitrary quantum algorithms. We show that even the powerful Demigod requires $NT = \Omega(\gamma/\theta^2)$ channel uses to distinguish the two channels. Thus, quantum agents, who are no more powerful than Demigod, require at least as many channel uses as Demigod to perform the distinguishing task.

We define the hypothetical agent Demigod to have the following enhanced capabilities:
\begin{itemize}
\item \textbf{Unlimited computational power:} Demigod can perform arbitrary classical or quantum computations instantaneously and with unlimited memory.
\item \textbf{Direct observation of rotation angles:} While quantum agents are given the ability to apply noisy quantum channels to quantum states, Demigod directly observes the noisy rotation angles, which are either $\theta_n$ or $\theta + \theta_n$ for signal $\theta$ and noise $\theta_n$.
\end{itemize}
More precisely, when the unknown channel is $\tilde{\mathcal{C}}_\theta^{\otimes N}$, each query provides Demigod with $N$ real numbers $\{\theta+\theta_{n,1}, \theta+\theta_{n,2}, \ldots, \theta+\theta_{n,N}\}$, where each $\theta_{n,i}$ is independently drawn from $\mathcal{N}(0,\gamma)$ with mean zero and variance $\gamma$. When the unknown channel is $\tilde{\mathcal{C}}_0^{\otimes N}$, she receives the list $\{\theta_{n,1}, \theta_{n,2}, \ldots, \theta_{n,N}\}$. 

Since Demigod can query the rotation angles, she can perfectly simulate any quantum agent. To simulate a query to the noisy quantum channel $\tilde{\mathcal{C}}_\theta^{\otimes N}$ or $\tilde{\mathcal{C}}_0^{\otimes N}$ made by the quantum agent, she uses her enhanced ability to directly observe the list of rotation angles and rotates each of her $N$ sensing qubits about the $Z$-axis by the corresponding observed angle. Due to her unlimited computational power, she can also simulate any quantum computation performed by the quantum agent. Thus, she can perfectly simulate any quantum agent regardless of the quantum state preparation, evolution, channel query, or measurement strategy employed by the original quantum algorithm. Therefore, Demigod is at least as powerful as any quantum agent who tries to distinguish the two noisy channels.

After $T$ queries to the parallel channel, Demigod has collected $NT$ rotation angles, each angle is independently and identically drawn from either $P = \mathcal{N}(0,\gamma)$ or $Q = \mathcal{N}(\theta,\gamma)$ depending on whether the noisy channel is $\tilde{\mathcal{C}}_0$ or $\tilde{\mathcal{C}}_\theta$, respectively. For Demigod, the channel distinguishing problem reduces to the classical statistical problem of distinguishing between these two Gaussian distributions.
The fundamental limit for this classical hypothesis testing problem is well-established in information theory. To distinguish between $P$ and $Q$ with constant success probability, the number of samples must satisfy:
\begin{align}
\text{Number of samples} = \Omega\left(\frac{1}{D_{\text{KL}}(P \| Q)}\right),
\end{align}
where $D_{\text{KL}}(P \| Q)$ is the Kullback-Leibler divergence. For Gaussian distributions:
\begin{align}
D_{\text{KL}}(\mathcal{N}(0,\gamma) \| \mathcal{N}(\theta,\gamma)) = \frac{\theta^2}{2\gamma}.
\end{align}
Therefore, even Demigod requires at least $\Omega(\gamma/\theta^2)$ samples, establishing $NT = \Omega(\gamma/\theta^2)$. This lower bound on the number of samples applies no matter how much classical or quantum computing power Demigod applies to her task. Taking the maximum of both bounds yields $NT = \Omega\left(\max\left( \frac{\gamma}{\theta^2}, \frac{1}{\theta} \right)\right)$. Thus, we obtain the desired lower bound on $NT$, completing the proof of Theorem \ref{thm:noisy_small_angle_detection}.
\end{proof}

\section{A Classical Heuristic for Simulating Quantum Circuits} \label{app:lowweight}

\subsection{Overview and Context}

The low-weight Pauli propagation algorithm~\cite{aharonov2023polynomial, shao2023simulating, nemkov2023fourier, beguvsic2023fast, beguvsic2023simulating, fontana2023classical, rudolph2023classical, angrisani2024classically, dowling2024magic} provides a powerful classical heuristic for simulating quantum circuits. The key idea is to truncate high-weight Pauli terms during the simulation, maintaining only the most important contributions. Consider a quantum circuit with $L$ layers:
\begin{equation}\label{eq:circuit}
    U = U_L U_{L-1} \cdots U_1,
\end{equation}
Our goal is to approximate expectation values of the form:
\begin{equation}\label{eq:expectation}
    f_U(O) = \Tr[U\rho U^\dagger O],
\end{equation}
where $\rho$ is an input state and $O$ is an observable.

\subsection{Algorithm Description}
The low-weight Pauli propagation algorithm proceeds in the Heisenberg picture, where the observable is evolved backward through the circuit.
\begin{enumerate}
\item We begin by initializing the operator $O_L$ as the projection of the observable $O$ onto the subspace of Pauli operators with weight at most $k$. Here, the weight of a Pauli operator is the number of non-identity terms it contains.
\item Working backwards from $j = L$ down to $1$, we iteratively apply the following two steps:
    \begin{enumerate}
        \item Evolve the operator through the $j$-th layer: $O'_{j-1} = U_j^\dagger O_j U_j$.
        \item Project the result back onto the low-weight Pauli basis to get the new operator:
        \begin{equation}\label{eq:propagate}
            O_{j-1} = \sum_{\substack{P \in \{I,X,Y,Z\}^{\otimes n}\\ |P| \leq k}} \frac{\Tr[O'_{j-1} P]}{2^n} P.
        \end{equation}
    \end{enumerate}
\item The final estimate for the expectation value is computed using the fully evolved and truncated operator $O_0$ and the initial state $\rho$:
    \begin{equation}
        \text{Estimate} = \Tr[O_0 \rho].
    \end{equation}
\end{enumerate}
The larger $k$ is, the more accurate the low-weight Pauli propagation algorithm is. In practice, one can typically consider $k$ to be a small constant. Even $k=1$ can perform quite well in many cases. This completes the full description of the low-weight Pauli propagation algorithm.

\subsection{Theoretical Guarantees}
When each circuit layer $U_j$ is drawn from a \emph{locally scrambling} distribution~\cite{kuo2020markovian, hu2021classical, caro2022outofdistribution, gibbs2024dynamical, huang2022learning}, which corresponds to a probability distribution that remains invariant under single-qubit Clifford rotations, \cite{angrisani2024classically} has established the following theorem to provide a rigorous guarantee for the low-weight Pauli propagation algorithm. Intuitively, the locally scrambling property means each layer is generic on a local scale and mixes the local basis.

\begin{theorem}[Time Complexity]\label{thm:complexity}
Let $U$ be sampled from an $L$-layered locally scrambling circuit ensemble on $n$ qubits. For any error tolerance $\epsilon$ and failure probability $\delta$, there exists a classical algorithm with runtime $L n^{\mathcal{O}(\log(\epsilon^{-1}\delta^{-1}))}$ that outputs an estimate $\alpha$ satisfying:
\begin{equation}\label{eq:error}
    |\alpha - f_U(O)| \leq \epsilon \|O\|_{\mathrm{Pauli},2}
\end{equation}
with probability at least $1-\delta$. $\|O\|_{\mathrm{Pauli},2} = (2^{-n}\Tr[O^\dagger O])^{1/2}$ is the normalized Hilbert-Schmidt norm.
\end{theorem}

This theorem establishes three key properties of the algorithm:
\begin{enumerate}
\item The runtime is polynomial in system size $n$ and circuit depth $L$ for any small constant $\epsilon, \delta$.
\item The accuracy parameters $\epsilon$ and $\delta$ appear only logarithmically in the exponent.
\item The error scales naturally with the observable's normalized Hilbert-Schmidt norm.
\end{enumerate}

\subsection{Applicability and Scope}

This theoretical result encompasses various quantum circuit architectures. For the computational complexity bounds to hold, we require that for any Pauli operator $P$ of weight $k$, its Heisenberg evolution $U_j^\dag P U_j$ contains at most $n^{\mathcal{O}(k)}$ distinct Pauli terms and can be computed classically in time $n^{\mathcal{O}(k)}$. Intuitively, this condition ensures each circuit layer $U_j$ remains reasonably shallow. The analysis applies to several important classes of quantum circuits:
\begin{enumerate}
\item Brickwork circuits of $SU(4)$ gates in 1D, 2D, and 3D~\cite{zhang2023absence, napp2022quantifying, braccia2024computing}.
\item Circuits with universal single-qubit rotations followed by entangling Clifford gates~\cite{letcher2023tight}.
\item Quantum Convolutional Neural Networks without feed-forward~\cite{pesah2020absence}.
\end{enumerate}
Despite this broad applicability, the fundamental limitations of this approach warrant careful consideration. The low-weight Pauli propagation method remains inherently a classical heuristic. Its theoretical guarantees extend only to \emph{most} quantum circuits where each layer is drawn from a locally scrambling distribution, rather than to arbitrary quantum circuits.

When presented with a specific quantum circuit of interest, determining whether this classical heuristic will yield accurate simulation results is a non-trivial challenge. This uncertainty points to a fundamental insight: as we demonstrate in the next section, the very task of determining whether a given quantum circuit exhibits an advantage over this classical heuristic constitutes a problem that itself admits a quantum advantage.

\section{Classical hardness in predicting quantum advantages}\label{app:hardness-predict-qadv}

There are many classical heuristics for simulating quantum circuits and quantum algorithms.
Some families of quantum circuits can be efficiently simulated using classical methods, nullifying any potential quantum advantage. However, the task of determining whether a given problem exhibits a quantum advantage can itself have a quantum advantage. Hence, even if a quantum circuit is classically simulatable using a specific classical heuristic, we might still need a quantum computer to determine or verify that the classical heuristic can indeed simulate the quantum circuit accurately.

\begin{definition}[A classical heuristic for simulating quantum circuits]
A classical algorithm $\mathcal{A}$ is called a classical heuristic for simulating a class of polynomial-size quantum circuits if:
\begin{enumerate}
    \item It runs in classical polynomial time.
    \item It correctly approximates the output probability $C(x)$ for a subclass of polynomial-size quantum circuits $C$, but it is known to fail, or is not proven to succeed, for some circuits in this class and generally the set of circuits for which it fails is unknown.
\end{enumerate}
The algorithm is considered a heuristic because there is a gap between its practical performance on many instances and the lack of a guarantee for universal correctness within its target problem class.
\end{definition}

A powerful example is the low-weight Pauli propagation heuristic, \textsc{LowWeightPauliProp}, discussed in Appendix~\ref{app:lowweight}. It is known to be effective for most circuits in a measure-theoretic sense (a large subclass) but is also known to fail for others, fitting our definition perfectly. Given that a classical heuristic may succeed for some circuits and fail for others, we need a precise, quantitative way to distinguish these cases. We formalize this by defining what it means for a quantum circuit to have a ``computational advantage'' over a classical heuristic. This occurs when the classical heuristic fails to provide a good approximation for a significant fraction of possible inputs.

\begin{definition}[Quantum advantage over a classical heuristic]
We say a quantum computer executing a quantum circuit $C$ exhibits a computational advantage over a classical heuristic $\mathcal{A}$ if
\begin{equation}
|\mathcal{A}_C(x) - C(x) | \geq 1/3,
\end{equation}
for at least a $2/3$ fraction of the inputs $x \in \{0, 1\}^n$. Conversely, we say it has no advantage if $|\mathcal{A}_C(x) - C(x)| < 1/3$ for at least a $2/3$ fraction of inputs.
\end{definition}

The specific constants $1/3$ and $2/3$ are chosen by convention, following standard practice in computational complexity theory to create a mathematically convenient definition. The threshold of $1/3$ for the error ensures that a heuristic's prediction is significantly wrong, not just slightly inaccurate. The requirement that this occurs for a $2/3$ fraction of inputs establishes that the failure is typical for the circuit, not a rare exception. Together, these values create a large, constant-sized gap between the ``advantage'' case, where the heuristic is wrong for a clear majority of inputs, and the ``no advantage'' case, where it is right for a clear majority.

In the following, we prove that deciding if a given quantum circuit yields a computational advantage over a classical heuristic is itself a problem with a quantum advantage, under the widely believed conjecture that there exist some quantum computations that are classically hard.

For simplicity, we focus on \textsc{LowWeightPauliProp} with $k = 1$. The statement holds for $k > 1$ using a similar proof.

\begin{theorem}[Quantum advantage in detecting quantum advantage]\label{prop:determineQAdv}
Consider the following computational problem for detecting quantum advantage in a given quantum circuit:
\begin{quote}
    \textsc{DetectingQuantumAdvantage}\\[0.3em]
    \textbf{Input:} A description of a quantum circuit $C$\\[0.1em]
    \textbf{Promise:} A quantum computer executing $C$ either exhibits a computational advantage over $\textsc{LowWeightPauliProp}$ or it does not (as defined above).\\[0.1em]
    \textbf{Output:} $\begin{cases}
        1, & \text{if $C$ exhibits a quantum advantage}\\
        0, & \text{if $C$ does not exhibit a quantum advantage}
    \end{cases}$
\end{quote}
\noindent This problem can be solved efficiently on a quantum computer (it is in $\mathsf{BQP}$) but is intractable for any classical algorithm (it is not in $\mathsf{BPP}$), assuming $\mathsf{BPP} \neq \mathsf{BQP}$.
\end{theorem}

\begin{proof}
We divide the proof into two parts. One focused on showing that the problem is quantumly easy. And the second one focused on showing that the problem is classically hard.

\vspace{0.5em}\paragraph*{The problem is in $\mathsf{BQP}$ (Quantumly Easy):}
A quantum algorithm can solve \textsc{DetectingQuantumAdvantage} efficiently.
\begin{enumerate}
    \item Given the circuit description $C$, choose a set of $s = \mathcal{O}(1)$ random input strings $\{x_1, \dots, x_s\}$.
    \item For each $x_i$ for $i = 1, \ldots, s$:
        \begin{enumerate}
            \item Compute the heuristic's prediction $\mathcal{A}_C(x_i)$ using the classical \textsc{LowWeightPauliProp} algorithm as described in Appendix~\ref{app:lowweight}.
            \item Estimate the true probability $C(x_i)$ on a quantum computer. This is done by preparing the state $C(\ket{x_i}\ket{0^m})$ and measuring the first qubit. Repeating this a polynomial number of times yields an estimate of $C(x_i)$ to within a small error $\delta \ll 1/3$ with high probability.
        \end{enumerate}
    \item For each $x_i$ for $i = 1, \ldots, s$, check if $|\mathcal{A}_C(x_i) - C(x_i)| \ge 1/3$.
    \item If the fraction of inputs $x_i$ satisfying this condition is greater than $1/2$, output~$1$. Otherwise, output~$0$. By a standard Chernoff bound with an appropriate sample size $s$, this sampling procedure correctly distinguishes the case where the true fraction is $\ge 2/3$ from the case where it is $\le 1/3$ with high probability.
\end{enumerate}
The entire procedure involves a polynomial number of classical steps and a polynomial number of runs on a quantum computer, so the problem is in $\mathsf{BQP}$.

\vspace{0.5em}\paragraph*{The problem is not in $\mathsf{BPP}$ (Classically Hard):} 
We show this by contradiction. Suppose there exists a polynomial-time classical algorithm $\mathcal{D}$ that can solve \textsc{DetectingQuantumAdvantage}. We will show how to use $\mathcal{D}$ to solve any problem in $\mathsf{BQP}$ classically, which would imply $\mathsf{BPP} = \mathsf{BQP}$, contradicting our assumption. 

Before diving into technicalities, we briefly explain the main idea underlying our argument. Given any polynomial-size quantum circuit $C_?$ that solves a decision problem in $\mathsf{BQP}$, we will construct a new polynomial-size quantum circuit $C_\text{new}$ with input $x$ that has these properties: (1) If the answer to the decision problem is YES, then the output of the classical heuristic disagrees with the output of the quantum circuit $C_\text{new}$ for most values of the input $x$. (2) If the answer to the decision problem is NO, then the output of the classical heuristic agrees with the output of the quantum circuit $C_\text{new}$ for most values of the input $x$. Therefore, the polynomial-time classical algorithm $\mathcal{D}$, by determining whether the classical heuristic is successful or not, solves the decision problem. 

Now we proceed with the formal proof. Let $C_?$ be a polynomial-size quantum circuit that solves an instance of an arbitrary problem in $\mathsf{BQP}$. The circuit $C_?$ takes the all-zero state $\ket{0^m}$ as input, and when we measure the first qubit of $C_? \ket{0^m}$ in the $Z$ basis, we are promised that the probability of observing $1$ is either $\geq 2/3$ (YES instance) or $\leq 1/3$ (NO instance).
We now construct a new circuit such that applying $\mathcal{D}$ to the new circuit will enable one to accurately distinguish between the YES and NO instances.

\vspace{0.5em}
\textit{Step 1: Random Circuit Construction.}
Let $U = U_L \cdots U_1$ be an $n$-qubit random quantum circuit where each $U_j$ is a layer of random two-qubit gates, and the total depth $L$ can be chosen arbitrarily as long as $L$ is polynomial in $n$ and $m$. By choosing $L$ to be at least linear in $n$ with a sufficiently large constant, it is known that $U$ forms an $\epsilon$-approximate unitary $2$-design with $\epsilon = \mathcal{O}(1/2^n)$ \cite{harrow2009random, brandao2016local, schuster2024random}. This means the distribution over the random unitary $U$ nearly matches the Haar measure over all $n$-qubit unitaries $\mathrm{SU}(2^n)$ up to the second moment.

\vspace{0.5em}
\textit{Step 2: Amplified Circuit Construction.}
To amplify the success probability, we construct $C_{?, \text{ext}}$ that runs $\ell = \mathcal{O}(\log n)$ independent copies of $C_?$ and performs a coherent majority vote:
\begin{enumerate}
\item Run $\ell$ copies of $C_?$ on separate ancilla registers.
\item Apply a coherent majority vote circuit to the first qubits of each copy, storing the result in a designated ancilla qubit $q_{\text{maj}}$.
\end{enumerate}
By standard concentration inequalities, such as the Chernoff bound, this amplified circuit satisfies:
\begin{itemize}
\item If $C_?$ is a YES instance: $\Pr[\text{Measuring} \,\, q_{\text{maj}} \,\, \text{in $Z$ basis gives} \,\, 1] \geq 1 - 2^{-\Omega(\ell)} \geq 1 - 1/\mathrm{poly}(n)$.
\item If $C_?$ is a NO instance: $\Pr[\text{Measuring} \,\, q_{\text{maj}} \,\, \text{in $Z$ basis gives} \,\, 1] \leq 2^{-\Omega(\ell)} \leq 1/\mathrm{poly}(n)$.
\end{itemize}
Note that $C_{?, \text{ext}}$ acts on $m\ell + 1$ qubits.

\vspace{0.5em}
\textit{Step 3: New Circuit Construction.}
We construct the circuit $C_{\text{new}}$ that operates on the main $n$ qubits plus all ancilla qubits:
\begin{equation}
C_{\text{new}} = (U_L \cdots U_1 \otimes \Id) \cdot \text{Controlled-}\left(U_1^\dagger \cdots U_L^\dagger \right) \cdot \left(\Id \otimes C_{?, \text{ext}} \right),
\end{equation}
where $\text{Controlled-}(U_1^\dagger \cdots U_L^\dagger)$ applies $U^\dagger = U_1^\dagger \cdots U_L^\dagger$ to the main $n$ qubits conditioned on the majority vote qubit $q_{\text{maj}}$ being in the state $\ket{1}$.

\vspace{0.5em}
Let our observable be the Pauli operator $O=Z_1$ acting on the first qubit of the main system, and our initial state over the main and the ancilla systems be $\rho = \ketbra{x}{x} \otimes \ketbra{0^{m\ell +1}}{0^{m\ell +1}}$ for an input $n$-bit string $x \in \{0, 1\}^n$. Let us analyze the true expectation value:
\begin{equation}
    \langle Z_1 \rangle_{\text{true}, x} = \Tr[O \cdot C_{\text{new}} \cdot \rho \cdot C_{\text{new}}^\dagger].
\end{equation}
We separate the analysis into two cases.

\vspace{0.5em}
\noindent \textit{Case 1: $C_?$ is a YES instance}
With probability $\geq 1 - 1/\mathrm{poly}(n)$, the majority vote qubit $q_{\text{maj}} = 1$, so the effective unitary applied to the main qubits is:
\begin{equation}
U_{\text{eff}} = U_L \cdots U_1 \cdot (U_1^\dagger \cdots U_L^\dagger) = I.
\end{equation}
Therefore, for any input bitstring $x \in \{0, 1\}^n$, we have the following,
\begin{align}
\left| \langle Z_1 \rangle_{\text{true}, x} - (-1)^{x_1} \right| \leq 1/\mathrm{poly}(n),
\end{align}
where $x_1$ denotes the first bit of the input string $x$.

\vspace{0.5em}
\noindent \textit{Case 2: $C_?$ is a NO instance}
With probability $\geq 1 - 1/\mathrm{poly}(n)$, the majority vote qubit $q_{\text{maj}} = 0$, so the effective unitary applied to the main qubits is just $U$. Therefore,
\begin{align}
\left| \langle Z_1 \rangle_{\text{true}, x} - \bra{x} U^\dagger Z_1 U \ket{x} \right| \leq 1/\mathrm{poly}(n).
\end{align}
Because $U$ forms an $1/\poly(n)$-approximate unitary $2$-design, we have $|\bra{x} U^\dagger Z_1 U \ket{x}| \leq 1 / \poly(n)$ with high probability over the choice of $U$. Hence, for any input bitstring $x \in \{0, 1\}^n$, we have
\begin{align}
\left| \langle Z_1 \rangle_{\text{true}, x} - 0 \right| \leq 1/\mathrm{poly}(n)
\end{align}
with high probability.

\vspace{0.5em}
Now, we analyze how the \textsc{LowWeightPauliProp} heuristic $\mathcal{A}$ simulates the new circuit $C_{\mathrm{new}}$ layer-by-layer. The \textsc{LowWeightPauliProp} heuristic evolves $Z_1$ backwards through the entire circuit.
The following lemma focuses on propagation over the main $n$-qubit system and the first $L$ layers in $C_{\mathrm{new}}$.

\begin{lemma}[Decay of Frobenius Norm] \label{lem:decay-heuristic}
Consider a main system of $n$ qubits and an ancillary system of $n'$ qubits. Let $U = U_L \cdots U_1$ be a random quantum circuit on the main $n$-qubit system, where each layer $U_j$ consists of $n/2$ Haar-random two-qubit gates acting on disjoint pairs of qubits. Let $O_{\mathrm{heur}}$ be the operator computed by evolving $Z_1$ acting on $n+n'$ qubits backward through $U \otimes \Id$ using the classical \textsc{LowWeightPauliProp} heuristic. Then we have
\begin{equation}
\mathbb{E}_U\left[\|O_{\mathrm{heur}}\|_F^2\right] = \left(\frac{2}{5}\right)^L \cdot 2^{n+n'}.
\end{equation}
\end{lemma}
\begin{proof}
We prove this by analyzing the evolution layer by layer. We track how the squared Frobenius norm changes as we apply each layer of the heuristic evolution. Let $O^{(L)} = Z_1$. For $j = L$ down to $1$, each step of the heuristic evolution is given by:
\begin{equation}
O^{(j-1)} = \mathcal{P}_1(U_j^\dagger O^{(j)} U_j),
\end{equation}
where $\mathcal{P}_1$ projects onto Pauli operators of weight at most 1. Consider a single layer $U_j$ consisting of Haar-random two-qubit gates on disjoint pairs. Because the unitary $U_j$ only acts on the first $n$ qubits and all Pauli operators that act on two or more qubits are truncated and the evolution preserves the trace of zero, the $(n+n')$-qubit operator $O^{(j)}$ can be written as:
\begin{equation}
O^{(j)} = \sum_{i=1}^n \sum_{P \in \{X,Y,Z\}} \alpha_{P_i}^{(j)} P_i,
\end{equation}
where $P_i$ acts as Pauli operator $P$ on the $i$-th qubit in the main $n$-qubit systems and acts as identity $\Id$ on all other qubits. The squared Frobenius norm is:
\begin{equation}
\|O^{(j)}\|_F^2 = 2^{n+n'} \sum_{i=1}^n \sum_{P \in \{X,Y,Z\}} |\alpha_{P_i}^{(j)}|^2.
\end{equation}
For $Z_1$, we have $\alpha_{Z_1} = 1$ and all other coefficients are zero, so $\|Z_1\|_F^2 = 2^{n+n'}$.

Since each qubit participates in exactly one two-qubit gate per layer, and the gates are independent Haar-random unitaries, we can formally compute the expected squared Frobenius norm after one layer. After applying a two-qubit gate layer $U_j$ and projecting with $\mathcal{P}_1$, we get:
\begin{equation}
O^{(j-1)} = \mathcal{P}_1(U_j^\dagger O^{(j)} U_j) = \sum_{i'=1}^n \sum_{P' \in \{X,Y,Z\}} \alpha_{P'_{i'}}^{(j-1)} P'_{i'},
\end{equation}
where the new coefficients are:
\begin{equation}
\alpha_{P'_{i'}}^{(j-1)} = \frac{1}{2^{n+n'}} \text{Tr}\left( P'_{i'} \cdot U_j^\dagger O^{(j)} U_j \right) = \frac{1}{2^{n+n'}} \sum_{i=1}^n \sum_{P \in \{X,Y,Z\}} \alpha_{P_i}^{(j)} \text{Tr}\left( P'_{i'} U_j^\dagger P_i U_j \right).
\end{equation}
The squared Frobenius norm is:
\begin{align}
\|O^{(j-1)}\|_F^2 &= 2^{n+n'} \sum_{i'=1}^n \sum_{P' \in \{X,Y,Z\}} |\alpha_{P'_{i'}}^{(j-1)}|^2 \\
&= 2^{n+n'} \sum_{i'=1}^n \sum_{P' \in \{X,Y,Z\}} \left| \frac{1}{2^{n+n'}} \sum_{i=1}^n \sum_{P \in \{X,Y,Z\}} \alpha_{P_i}^{(j)} \text{Tr}\left( P'_{i'} U_j^\dagger P_i U_j \right) \right|^2 \\
&= \frac{1}{2^{n+n'}} \sum_{i'=1}^n \sum_{P' \in \{X,Y,Z\}} \left| \sum_{i=1}^n \sum_{P \in \{X,Y,Z\}} \alpha_{P_i}^{(j)} \text{Tr}\left( P'_{i'} U_j^\dagger P_i U_j \right) \right|^2
\end{align}
Taking expectation over the random two-qubit gate layer $U_j$:
\begin{align}
&\mathbb{E}_{U_j}[\|O^{(j-1)}\|_F^2] \\
&= \frac{1}{2^{n+n'}} \sum_{i'=1}^n \sum_{P' \in \{X,Y,Z\}} \mathbb{E}_{U_j} \left[ \left| \sum_{i=1}^n \sum_{P \in \{X,Y,Z\}} \alpha_{P_i}^{(j)} \text{Tr}\left( P'_{i'} U_j^\dagger P_i U_j \right) \right|^2 \right] \\
&= \frac{1}{2^{n+n'}} \sum_{i'=1}^n \sum_{i,i''=1}^n \sum_{\substack{P \in \{X,Y,Z\}\\P' \in \{X,Y,Z\}\\P'' \in \{X,Y,Z\}}} \alpha_{P_i}^{(j)} \overline{\alpha_{P''_{i''}}^{(j)}} \mathbb{E}_{U_j} \left[ \text{Tr}\left( P'_{i'} U_j^\dagger P_i U_j \right) \overline{\text{Tr}\left( P'_{i'} U_j^\dagger P''_{i''} U_j \right)} \right]\\
&= \frac{1}{2^{n+n'}} \sum_{i'=1}^n \sum_{i,i''=1}^n \sum_{\substack{P \in \{X,Y,Z\}\\P' \in \{X,Y,Z\}\\P'' \in \{X,Y,Z\}}} \alpha_{P_i}^{(j)} \overline{\alpha_{P''_{i''}}^{(j)}} \mathbb{E}_{U_j} \left[ \text{Tr}\left( P'_{i'} U_j^\dagger P_i U_j \right) \text{Tr}\left(  P'_{i'} U_j^\dagger P''_{i''} U_j \right) \right].
\end{align}
Now we need to compute the expectation $\mathbb{E}_{U_j} \left[ \text{Tr}\left( P'_{i'} U_j^\dagger P_i U_j \right) \text{Tr}\left(  P'_{i'} U_j^\dagger P''_{i''} U_j \right) \right]$. Since each qubit participates in exactly one two-qubit gate per circuit layer, we have two cases:
\begin{itemize}
    \item \textbf{Case 1:} Qubits $i$ and $i''$ are involved in different two-qubit gates. Let gate $V_1$ act on qubits $\{i, j\}$ and gate $V_2$ act on qubits $\{i'', j''\}$ where $\{i,j\} \cap \{i'',j''\} = \emptyset$. Since $V_1$ and $V_2$ are independent Haar-random gates, and $P_i$ only acts on the space of $V_1$ while $P''_{i''}$ only acts on the space of $V_2$,
    \begin{equation}
    \mathbb{E}_{V_1,V_2} \left[ \text{Tr}\left( P'_{i'} V_1^\dagger P_i V_1 \right) \text{Tr}\left( P'_{i'} V_2^\dagger P''_{i''} V_2 \right) \right] = 0.
    \end{equation}
    \item \textbf{Case 2:} Qubits $i$ and $i''$ are involved in the same two-qubit gate that act on the two qubits $\{i, j\}$. For a Haar-random two-qubit unitary gate $V$ and single-qubit Pauli operators $P_i$, $P''_{i''}$, $P'_{i'}$ acting within the same two-qubit space $\{i, j\}$, we have
    \begin{equation}
    \mathbb{E}_V\left[ \text{Tr}(P'_{i'} V^\dagger P_i V) \text{Tr}(P'_{i'} V^\dagger P''_{i''} V) \right] = \frac{4^{n+n'}}{15} \cdot \delta_{P_i, P''_{i''}},
    \end{equation}
    from the standard Weingarten calculation for Haar-random unitaries \cite{harrow2009random, brandao2016local}. Here, $\delta_{P_i, P''_{i''}}$ is $1$ if the two $n$-qubit Pauli operators $P_i, P''_{i''}$ are the same and $0$ otherwise.
    If $i'$ is not acted on by the two-qubit gate $V$, then because $P'_{i'}$ is traceless, we have
    \begin{equation}
        \mathbb{E}_V[\text{Tr}(P'_{i'} V^\dagger P_i V) \text{Tr}(P'_{i'} V^\dagger P''_{i''} V)] = 0.
    \end{equation}
\end{itemize}
Now, collecting all the surviving terms in the sum, we have
\begin{align}
\mathbb{E}_{U_j}[\|O^{(j-1)}\|_F^2] &= \frac{4^{n+n'}}{2^{n+n'}} \sum_{i=1}^n \sum_{P \in \{X, Y, Z\}} 6 \cdot \frac{1}{15} \cdot \left|\alpha_{P_i}^{(j)}\right|^2,
\end{align}
where $6$ comes from the possible choices of $P'_{i'}$ such that qubit ${i'}$ is involved in the unique two-qubit gate in $U_j$ that acts on qubit $i$. Using $\|O^{(j)}\|_F^2 = 2^{n+n'} \sum_{i=1}^n \sum_{P \in \{X,Y,Z\}} |\alpha_{P_{i}}^{(j)}|^2$, we have
\begin{equation}
\mathbb{E}_{U_j}[\|O^{(j-1)}\|_F^2] = \frac{2}{5} \|O^{(j)}\|_F^2.
\end{equation}
After $L$ layers, we have
\begin{equation}
\mathbb{E}_U\left[\|O_{\mathrm{heur}}\|_F^2\right] = \mathbb{E}_{U}[\|O^{(0)}\|_F^2] = \left(\frac{2}{5}\right)^L \mathbb{E}_{U}[\|O^{(L)}\|_F^2] = \left(\frac{2}{5}\right)^L \|Z_1\|_F^2 = \left(\frac{2}{5}\right)^L \cdot 2^{n+n'}.
\end{equation}
This concludes the proof.
\end{proof}

Let $\mathcal{P}$ be the projection onto the subspace of weight-$1$ Pauli operators. Let $T_V(O) = \mathcal{P}(V^\dagger O V)$ be the map for one step of the heuristic's backward evolution through a unitary $V$.

\begin{lemma}[Non-increasing Frobenius norm under the heuristic evolution]\label{lem:nonincreasing}
For any unitary $V$ acting on $q$ qubits and any operator $O$ on the same space, we have
\begin{equation}
\|T_V(O)\|_F \leq \|O\|_F.
\end{equation}
\end{lemma}
\begin{proof}
Since $V$ is unitary, conjugation by $V$ preserves the Frobenius norm: $\|V^\dagger O V\|_F = \|O\|_F$. The projection $\mathcal{P}$ onto a subspace can only decrease the Frobenius norm, so $\|\mathcal{P}(A)\|_F \leq \|A\|_F$ for any operator $A$. Therefore, we have the following,
\begin{equation}
\|T_V(O)\|_F = \|\mathcal{P}(V^\dagger O V)\|_F \leq \|V^\dagger O V\|_F = \|O\|_F.
\end{equation}
This completes the proof.
\end{proof}

Now we analyze the classical heuristic's simulation of circuit $C_{\text{new}}$. The circuit $C_{\text{new}}$ acts on the combined Hilbert space of $n$ main qubits and $m\ell + 1$ ancilla qubits, for a total of $n + m\ell + 1$ qubits. The observable of interest is $Z_1 \otimes \Id^{\otimes (m\ell + 1)}$, which we denote as $Z_1^{\text{ext}}$ for the extended system. The classical heuristic evolves $Z_1^{\text{ext}}$ backwards through the circuit $C_{\text{new}} = (U_L \otimes \Id) \cdots (U_1 \otimes \Id) \cdot \text{Controlled-}(U^\dagger) \cdot (\Id \otimes C_{?, \text{ext}})$. We analyze this evolution in reverse order:
\begin{itemize}
    \item \textit{Backwards evolution through the random circuit layers $(U_L \otimes \Id) \cdots (U_1 \otimes \Id)$:} We start by evolving $Z_1^{\text{ext}}$ backwards through the $L$ layers of random gates. Each layer applies $T_{U_j \otimes \Id}$ to the current operator. Let $O_L = Z_1^{\text{ext}} = Z_1 \otimes \Id^{\otimes (m\ell + 1)}$. After backwards evolution through $j = L$ down to $1$:
    \begin{equation}
    O_{j-1} = T_{U_j \otimes \Id}(O_{j}) = \mathcal{P}((U_j^\dagger \otimes \Id) O_{j} (U_j \otimes \Id)).
    \end{equation}
    Because the random circuit layers $U_j$ only act on the main $n$ qubits, we have
    \begin{equation}
        O_{\text{post-random}} := O_0 = O_{\mathrm{heur}}
    \end{equation}
    for the $(n+n')$-qubit observable $O_{\mathrm{heur}}$ defined in Lemma~\ref{lem:decay-heuristic} with $n' = m\ell+1$. We have:
    \begin{equation}
    \mathbb{E}_U[\|O_{\text{post-random}}\|_F^2] = \mathbb{E}_U[\|O_{\mathrm{heur}}\|_F^2] = \left(\frac{2}{5}\right)^L \cdot 2^{n + m\ell + 1}.
    \end{equation}
    This implies that over the choices of the $L$-layer random quantum circuit $U$, the average Frobenius norm decays exponentially in $L$.

    \item \textit{Backwards evolution through the controlled operation:} The controlled operation is:
    \begin{equation}
    \text{Controlled-}(U^\dagger) = |0\rangle\langle 0|_{q_{\text{maj}}} \otimes \Id_{\text{main}} \otimes \Id_{\text{rest}} + |1\rangle\langle 1|_{q_{\text{maj}}} \otimes U_{\text{main}}^\dagger \otimes \Id_{\text{rest}}.
    \end{equation}
    After backwards evolution, we have the observable,
    \begin{equation}
    O_{\text{post-controlled}} = T_{\text{Controlled-}(U^\dagger)}(O_{\text{post-random}}).
    \end{equation}
    By Lemma~\ref{lem:nonincreasing}, we have $\|O_{\text{post-controlled}}\|_F \leq \|O_{\text{post-random}}\|_F$.

    \item \textit{Backwards evolution through $(\Id \otimes C_{?, \text{ext}})$:} Finally, we evolve backwards through $(\Id \otimes C_{?, \text{ext}})$:
    \begin{equation}
    O_{\text{final}} = T_{\Id \otimes C_{?, \text{ext}}}(O_{\text{post-controlled}}).
    \end{equation}
    By Lemma~\ref{lem:nonincreasing}, we have $\|O_{\text{final}}\|_F \leq \|O_{\text{post-controlled}}\|_F$.
\end{itemize}

Combining these inequalities, we obtain the following:
\begin{equation}
\mathbb{E}_U[\|O_{\text{final}}\|_F^2] \leq \mathbb{E}_U[\|O_{\text{post-controlled}}\|_F^2] \leq \mathbb{E}_U[\|O_{\text{post-random}}\|_F^2] \leq \left(\frac{2}{5}\right)^L \cdot 2^{n + m\ell + 1}.
\end{equation}
By Markov's inequality, with probability at least $1 - 2^{-n}$ over the choice of $U$, we have
\begin{equation}
    \|O_{\text{final}}\|_F^2 \leq \left(\frac{2}{5}\right)^L \cdot 2^{2n + m\ell + 1}.
\end{equation}
Therefore, for any input bitstring $x \in \{0, 1\}^n$, we have
\begin{equation}
|\langle Z_1 \rangle_{\text{heur}, x}| = |\text{Tr}[O_{\text{final}} \cdot (|x\rangle\langle x| \otimes |0^{m\ell+1}\rangle\langle 0^{m\ell+1}|)]| \leq \|O_{\text{final}}\|_F \leq \sqrt{\left(\frac{2}{5}\right)^L \cdot 2^{2n + m\ell + 1}}.
\end{equation}
We can choose $L$ arbitrarily as long as $L$ is polynomial in $n + m \ell + 1$. By choosing $L$ to be greater than $6(n + m \ell + 1)$, we have $\left(\frac{2}{5}\right)^L \cdot 2^{2 n + m\ell + 1} \leq (1/4)^{n + m \ell + 1}$. Hence, irrespective of the YES and NO instances, the output of the classical heuristic for any input bitstring $x \in \{0, 1\}^n$ satisfies
\begin{equation}
    |\langle Z_1 \rangle_{\text{heur}, x}| \leq \left(\frac{1}{2}\right)^{n + m \ell + 1}
\end{equation}
for all except an exponentially small fraction of random quantum circuits $U$.

Now we establish the separation between YES and NO instances. For YES instances, for any input bitstring $x \in \{0, 1\}^n$, we have:
\begin{equation}
|\langle Z_1 \rangle_{\text{true}, x} - \langle Z_1 \rangle_{\text{heur}, x}| \geq |(-1)^{x_1}| - |\langle Z_1 \rangle_{\text{heur}, x}| - \frac{1}{\text{poly}(n)} \geq 1 - \frac{1}{\text{poly}(n)} > \frac{2}{3}
\end{equation}
for any sufficiently large $n$ and all except exponentially small fraction of random quantum circuit $U$. For NO instances, for any input bitstring $x \in \{0, 1\}^n$, we have:
\begin{equation}
|\langle Z_1 \rangle_{\text{true}, x} - \langle Z_1 \rangle_{\text{heur}, x}| \leq |\langle Z_1 \rangle_{\text{true}, x}| + |\langle Z_1 \rangle_{\text{heur}, x}| \leq \frac{1}{\text{poly}(n)} < \frac{1}{3}
\end{equation}
for any sufficiently large $n$ and all except exponentially small fraction of random quantum circuit $U$. This establishes the following:
\begin{itemize}
\item If $C_?$ is a YES instance, then a quantum computer executing quantum circuit $C_{\text{new}}$ exhibits a computational advantage over \textsc{LowWeightPauliProp}, so $\mathcal{D}(C_{\text{new}}) = 1$.
\item If $C_?$ is a NO instance, then a quantum computer executing quantum circuit $C_{\text{new}}$ does not exhibit a computational advantage over \textsc{LowWeightPauliProp}, so $\mathcal{D}(C_{\text{new}}) = 0$.
\end{itemize}
Therefore, we can solve the original $\mathsf{BQP}$ problem for $C_?$ by the following classical algorithm:
\begin{enumerate}
\item Construct the circuit $C_{\text{new}}$ from $C_?$ as described above.
\item Apply the assumed classical algorithm $\mathcal{D}$ to determine if a quantum computer executing quantum circuit $C_{\text{new}}$ does not exhibit a computational advantage over \textsc{LowWeightPauliProp}.
\item Output the result of $\mathcal{D}$.
\end{enumerate}
Since this procedure works for any $\mathsf{BQP}$ problem and runs in polynomial time, we would have $\mathsf{BPP} = \mathsf{BQP}$, contradicting our assumption that $\mathsf{BPP} \neq \mathsf{BQP}$. Therefore, no polynomial-time classical algorithm can solve \textsc{DetectingQuantumAdvantage}, completing the proof.
\end{proof}

\clearpage
\bibliography{refs}

\end{document}